\documentclass[11pt]{article}
\usepackage[utf8]{inputenc}
\usepackage{mathtools}
\usepackage{amssymb,amsmath,amsthm,graphicx,color}
\usepackage{physics}
\usepackage{xspace}
\usepackage{enumitem}
\usepackage{authblk}
\pdfoutput=1
\setlist[itemize]{nosep,topsep=3pt,itemsep=3pt}
\setlist[enumerate]{nosep,topsep=3pt,itemsep=3pt}

\usepackage[dvipsnames]{xcolor}
\usepackage{comment}
\usepackage{svg}
\usepackage{multirow}
\usepackage[colorlinks=true,linktoc=page]{hyperref}
\usepackage{cleveref}

\usepackage[compact]{titlesec}

\usepackage{times}
\normalfont
\usepackage[T1]{fontenc}

\usepackage[american]{babel}

\usepackage{float}
\usepackage{mathpazo}

\usepackage[varg]{pxfonts} %
\usepackage[full]{textcomp}
\usepackage[scr=rsfso]{mathalfa}%
\usepackage{bm} %

\usepackage{libertine}
\usepackage[T1]{fontenc}

\usepackage{amsthm,amsfonts,framed}
\usepackage[margin=1in]{geometry}
\usepackage{physics, graphicx}

\newcommand{\invpolyn}{\frac{1}{\poly(n)}}
\newcommand{\invexpn}{\frac{1}{2^{\poly(n)}}}
\newcommand{\oracletv}{\mathcal{\widetilde{O}}_{T_V}}
\newcommand{\oracletnone}{\mathcal{\widetilde{O}}_{T_\varnothing}}
\newcommand{\symsubtv}{\mathsf{\widetilde{V}}_{T_V}}
\newcommand{\symsubtnone}{\mathsf{\widetilde{V}}_{T_\varnothing}}
\newcommand{\problemexpansionnoarg}{\textsf{\textup{NON-EXPANSION}}\ensuremath{(d,\alpha,\epsilon)}}
\newcommand{\problemexpansion}[2][d,\alpha,\epsilon]{\textsf{\textup{#2 NON-EXPANSION}}\ensuremath{(#1)}}
\newcommand{\problemhiddensubset}[1]{\textsf{\textup{#1 HIDDEN SUBSET}}\ensuremath{(\alpha)}}
\newcommand{\problemhiddensubgroup}[1]{\textsf{\textup{#1 HIDDEN SUBGROUP}}\ensuremath{(\times, \{H_i\})}}

\usepackage[compact]{titlesec}
\titlespacing{\section}{0pt}{1.5ex}{0ex}
\titlespacing{\subsection}{0pt}{1.5ex}{0ex}
\titlespacing{\subsubsection}{0pt}{1ex}{0ex}
\titlespacing{\paragraph}{0pt}{1.5ex}{1ex}
\usepackage{slantsc}

\newcommand{\plus}{\ket{+}^{\otimes n}}
\newcommand{\braplus}{\bra{+}^{\otimes n}}

\newcommand{\pmset}{\{\pm 1\}}

\hypersetup{colorlinks,
    linkcolor=blue,
    citecolor=ForestGreen,  
    urlcolor=Mahogany,
}

\usepackage{microtype}

\newcommand{\handout}[5]{
   \renewcommand{\thepage}{#1-\arabic{page}}
   \noindent
   \begin{center}
   \framebox{
      \vbox{
    \hbox to 5.78in { {\bf #1}
     	 \hfill #2 }
       \vspace{4mm}
       \hbox to 5.78in { {\Large \hfill #5  \hfill} }
       \vspace{2mm}
       \hbox to 5.78in { {\it #3 \hfill #4} }
      }
   }
   \end{center}
   \vspace*{4mm}
}

\newenvironment{proof-sketch}{\noindent{\bf Sketch of Proof:}\hspace*{1em}}{\qed\bigskip}
\newenvironment{proof-idea}{\noindent{\bf Proof Idea}\hspace*{1em}}{\qed\bigskip}
\newenvironment{proof-of-lemma}[1]{\noindent{\bf Proof of Lemma #1}\hspace*{1em}}{\qed\bigskip}
\newenvironment{proof-attempt}{\noindent{\bf Proof Attempt}\hspace*{1em}}{\qed\bigskip}

\makeatletter
\def\fnum@figure{{\bf Figure \thefigure}}
\def\fnum@table{{\bf Table \thetable}}
\long\def\@mycaption#1[#2]#3{\addcontentsline{\csname
  ext@#1\endcsname}{#1}{\protect\numberline{\csname
  the#1\endcsname}{\ignorespaces #2}}\par
  \begingroup
    \@parboxrestore
    \small
    \@makecaption{\csname fnum@#1\endcsname}{\ignorespaces #3}\par
  \endgroup}
\def\mycaption{\refstepcounter\@captype \@dblarg{\@mycaption\@captype}}
\makeatother

\newcommand{\mathify}[1]{\ifmmode{#1}\else\mbox{$#1$}\fi}
\newcommand{\bigO}O

\newcommand{\remove}[1]{}
\newcommand{\ignore}[1]{}

\def\pmset{\{\pm 1\}}

\def\poly{{\rm poly}}

\newcommand{\complexityclass}[1]{{\mathsf{#1}}\xspace}

\newcommand{\BQP}{\complexityclass{BQP}}

\newcommand{\MA}{\complexityclass{MA}}

\newcommand{\QMA}{\complexityclass{QMA}}
\newcommand{\QCMA}{\complexityclass{QCMA}}

\newcommand{\E}{\mathop{\bf E\/}}

\usepackage{amsthm}
\usepackage{thmtools,thm-restate}

\makeatletter
\renewcommand{\paragraph}{%
  \@startsection{paragraph}{4}%
  {\z@}{1.75ex \@plus 1ex \@minus .2ex}{-1em}%
  {\normalfont\normalsize\bfseries}%
}
\makeatother

\numberwithin{equation}{section}
\newcounter{t}
\numberwithin{t}{section}

\declaretheoremstyle[bodyfont=\it,qed=\qedsymbol,headpunct=.\vphantom{$p_{p_{p_p}}$},postheadspace=\newline,headformat=\NAME\  \NUMBER\,\NOTE]{noproofstyle} 

\newtheoremstyle{break}%
{}{}%
{\itshape}{}%
{\bfseries}{.\vphantom{$p_{p_{p_p}}$}}%
{\newline}
{\thmname{#1}\thmnumber{ #2}\thmnote{\ \,\textmd{(#3)}}}

\theoremstyle{break}

\declaretheorem[name=Observation,numbered=no]{observation*}

\declaretheorem[numberlike=t]{fact}

\declaretheorem[numberlike=t]{problem}

\declaretheorem[numberlike=t]{theorem}

\declaretheorem[name=Theorem,numbered=no]{theorem*}

\declaretheorem[numberlike=t]{lemma}
\declaretheorem[name=Lemma,numbered=no]{lemma*}

\declaretheorem[name=Corollary,numbered=no]{corollary*}

\declaretheorem[name=Parameter,numbered=no]{parameter*}

\declaretheorem[name=Proposition,numbered=no]{proposition*}

\declaretheorem[name=Claim,numbered=no]{claim*}

\declaretheorem[name=Conjecture,numbered=no]{conjecture*}

\declaretheorem[name=Question,numbered=no]{question*}

\declaretheoremstyle[bodyfont=\it,headpunct=.\vphantom{$p_{p_{p_p}}$},postheadspace=\newline,headformat=\NAME\  \NUMBER\,\NOTE]{defstyle} 

\declaretheorem[numberlike=t,style=defstyle]{definition}

\declaretheorem[unnumbered,name=Example,style=defstyle]{example*}

\declaretheorem[unnumbered,name=Notation=defstyle]{notation*}

\declaretheorem[unnumbered,name=Construction,style=defstyle]{construction*}

\declaretheorem[numberlike=t]{remark}
\declaretheorem[numberlike=t]{procedure}

\begin{document}
\author[]{Roozbeh Bassirian\thanks{\href{mailto:roozbeh@uchicago.edu}{roozbeh@uchicago.edu}}}
\author[]{Bill Fefferman\thanks{\href{mailto:wjf@uchicago.edu}{wjf@uchicago.edu}}}
\author[]{Kunal Marwaha\thanks{\href{mailto:kmarw@uchicago.edu}{kmarw@uchicago.edu}}}
\affil{University of Chicago}

\title{On the power of nonstandard quantum oracles}
\date{}
\maketitle
\vspace{-10mm}
\begin{abstract}
We study how the choices made when designing an oracle affect the complexity of quantum property testing problems defined relative to this oracle.
We encode a regular graph of even degree as an invertible function $f$, and present $f$ in different oracle models. We first give a one-query $\QMA$ protocol to test if a graph encoded in $f$ has a small disconnected subset.
We then use representation theory to show that no classical witness can help a quantum verifier efficiently decide this problem relative to an in-place oracle.
Perhaps surprisingly, a simple modification to the standard oracle prevents a quantum verifier from efficiently deciding this problem, even with access to an unbounded witness.
\end{abstract}


\pagenumbering{arabic}

\section{Introduction}
Computational complexity is the study of the innate amount of resources required to complete some task.
We assign \emph{complexity classes} to sets of tasks that require similar amounts of resources; from here, the goal is to understand the relationship between complexity classes.
There has been some success proving that two complexity classes are equal, for example $\text{IP}=\text{PSPACE}$~\cite{shamir1992ip}, the PCP theorem~\cite{arora1998proof}, and $\text{MIP}^*=\text{RE}$~\cite{ji2021mip}; however, proving that two complexity classes are \emph{unequal} has been much more elusive. For example, we cannot prove $\text{P} \ne \text{PSPACE}$, let alone $\text{P} \ne \text{NP}$.

One response to this difficulty is to equip a computational model with an \emph{oracle}, which computes a fixed (but arbitrarily powerful) quantity in a single timestep.
It is often easier to prove that a statement (e.g. $\text{P} \ne \text{NP}$) is true relative to an oracle; furthermore, this 
restricts the kinds of proof techniques that can show the statement is false without an oracle.
In addition to separating complexity classes, oracles and \emph{query complexity} naturally arise in cryptography (e.g.~\cite{katz2020introduction}) and learning theory (e.g. \cite{KM93}).

Even with respect to an oracle, proving that some complexity classes are unequal can be surprisingly difficult.
Notably, Aharonov and Naveh define $\QCMA$, a subset of $\QMA$ where the witness is a classical bitstring~\cite{aharonovnaveh}, and ask if $\QCMA \subsetneq \QMA$. 
Aaronson and Kuperberg conjecture that an oracle separates these classes, but only prove a ``quantum oracle'' where this occurs~\cite{Aaronson2007}.
Subsequent works~\cite{fefferman2018quantum,nn_2022_classical_randomized} remove the ``quantumness'' from the oracle model, but still use models with internal randomness or other nonstandard aspects.

We consider quantum property testing problems defined relative to oracles from various oracle models:
encoding the edges of a graph in an invertible function $f$, we present $f$ as either a \emph{standard} oracle or \emph{in-place} oracle, with or without internal randomness. With mild restrictions on the workspace of quantum verifiers, we find:
\begin{enumerate}
    \item In several oracle models presenting $f$, a \emph{quantum} witness can help a quantum verifier efficiently decide if the graph encoded in $f$ has a small disconnected subset.
    \item Where $f$ is presented as a randomized in-place oracle, no \emph{classical} witness can help a quantum verifier efficiently decide this problem.
    \item Where $f$ is presented as a randomized phase oracle, no witness \emph{of any type or size} can help a quantum verifier efficiently decide this problem.
\end{enumerate}
Our results highlight that the quantum complexity of a task defined relative to an oracle is influenced by the choice of oracle model.

\subsection{Our techniques}
We use a well-known fact of Petersen to encode the edges of any even-degree regular graph in an invertible function $f$. 
We then consider natural ways to install $f$ within an oracle; we say that $f$ is \emph{presented} as a particular kind of oracle.
For example, a standard oracle presents $f$ through the map $\ket{c,x} \mapsto \ket{c\oplus f(x),x}$, while an in-place oracle presents $f$ through the map $\ket{x} \mapsto \ket{f(x)}$. In general, we consider oracles that give access both to $f$ and $f^{-1}$. An oracle may also have internal randomness: on every query to a \emph{randomized} oracle, $f$ is chosen uniformly at random from a fixed set of functions $F$.

Consider the Laplacian $L_f$ of a graph encoded in $f$. We first provide a test such that for any input state $\ket{\psi}$, the test succeeds with probability expressible in terms of $\bra{\psi}L_f\ket{\psi}$, independently of how an oracle presents $f$.
We use this test to construct a $\QMA$ protocol verifying that the graph is not an \emph{expander} graph.
This problem is primarily motivated by the preimage-testing problem of Fefferman and Kimmel~\cite{fefferman2018quantum}, which separates $\QMA$ and $\QCMA$ relative to a nonstandard oracle. 
They encode an invertible function $\pi$ in an oracle \emph{without efficient access to $\pi^{-1}$}, and test a property of $\pi^{-1}$; by design, this property can be verified but not easily computed.
Crucially, we view a permutation and its inverse as the edges of an \emph{undirected graph}; properties of undirected graphs are not sensitive to the ordering of $(x, \pi(x))$. 
We use multiple permutations to study graphs of higher degree, and notice that detecting if a graph has a non-expanding region is hard without traversing most of the graph. 
Some of these ideas are related to the \emph{component mixer} concept of Lutomirski~\cite{lutomirski2011component}, and are simultaneously and independently explored by Natarajan and Nirkhe~\cite{nn_2022_classical_randomized}.

A randomized oracle presenting a set of functions $F$ can be seen as a quantum channel, so small changes to $F$ cause statistically indistinguishable changes to the oracle.
We use this flexibility to modify non-expansion testing to a simple permutation problem: do the functions $f \in F$ stabilize a small set $V \subseteq [N]$, or is $F$ the set of all permutations on $[N]$? 
Notice that $F$ is a group in both cases.
When an oracle presenting $F$ preserves the group structure of $F$, we can use representation theory.
For this problem, this is satisfied by an \emph{in-place} oracle; the oracle is then an orthogonal projector onto one of two symmetric subspaces of matrices.
After finding an orthogonal basis for each subspace, we construct a hybrid argument to prove that only witnesses with knowledge of $V$ can help a quantum verifier efficiently decide this problem. We also use representation theory to give a $\QCMA$ protocol for an analogous permutation problem in randomized standard oracles.

We finally study the permutation problem in a randomized phase oracle. 
We directly analyze the effect of the oracle on any input density matrix; with high probability, the oracle decreases the magnitude of every off-diagonal term by a $\invexpn$ factor.
We then construct a hybrid argument, bounding our measure of progress using an inequality relating the sizes of Schatten $p$-norms.
When the state space is not too large, we prove that an exponential number of queries are required to distinguish most YES instances from the NO instance, \emph{regardless of input state}. 
As a result, no witness can help a verifier distinguish YES from NO.

Note that our quantum verifiers are not fully general. 
Our lower bound techniques restrict the number of extra workspace qubits in the verifier; however, our upper bounds also work in this setting.
In \Cref{sec:classical_witness_lower_bound}, we explain these restrictions in more detail and discuss prospects for generalizing our results. 

\subsection{Related work}

\paragraph{Quantum oracle models}
A fundamental constraint of quantum oracle models is that they must be unitary. We describe several nonstandard oracle models used in quantum computing:
\begin{itemize}
    \item A \emph{quantum} oracle is any unitary operation $U$ in the full Hilbert space. Although the operation is unitary, the verifier doesn't necessarily have access to $U^{-1}$.  Oracles like these are not typically classical because the unitary's action is not efficiently and classically representable. 
    \item An \emph{in-place} oracle maps $\ket{x} \to \ket{\pi(x)}$ for some classical invertible function $\pi$. Again, this computation is not efficiently reversible since the verifier may not have access to $\pi^{-1}$. When a standard oracle gives access to $\pi^{-1}$, an in-place oracle query can be simulated in two queries; otherwise, an exponential number of queries are required to construct one from the other~\cite{Kashefi_2002}.
    \item A \emph{phase} oracle puts the output of a classical function $f$ in the phase of a basis state. We consider the map $\ket{x} \to e^{f(x) \cdot 2\pi i/N} \ket{x}$. To contrast, note that the map $\ket{c,x} \to e^{cf(x) \cdot 2\pi i/N}\ket{c,x}$ is unitarily equivalent to the standard oracle.
\end{itemize}

All of these oracles can optionally have \emph{internal randomness}, as considered by Harrow and Rosenbaum~\cite{harrow_uselessness}; we call these \emph{randomized} oracles. On every query to a randomized oracle, a unitary is chosen at random from a fixed set. This can be very powerful; for example,~\cite{harrow_uselessness} gives examples of randomized oracles where problems \emph{impossible} to decide with classical queries can be decided with a single quantum query.

\paragraph{$\QMA$ and $\QCMA$}
The \emph{Merlin-Arthur} style of complexity classes considers a decision problem and two players. The magician (Merlin) has claimed the answer to the decision problem is YES, and gives the verifier a token (the \emph{proof} or \emph{witness}) to convince them. The verifier (Arthur) must then ensure the answer is actually YES.
Given a problem with size $n$, the verifier must accept a correct witness (i.e.  when the answer is YES) with probability $1/q$ higher than a ``lying'' witness (i.e. when the answer is NO) for some $q = \poly(n)$.
The set of problems that can be decided this way in a classical setting is known as Merlin Arthur ($\MA$).
If the verifier is a quantum computer, this is $\QCMA$; if the witness can be any quantum state, this is $\QMA$.

\begin{table}[h]
\setlength{\tabcolsep}{1em}
\def\arraystretch{1.5}
\centering
\begin{tabular}{l|c c}
                            & \textbf{verifier is classical} & \textbf{verifier is quantum} \\ \hline
\textbf{witness is classical}   & $\MA$                        & $\QCMA$                          \\
\textbf{witness is quantum} & -                           & $\QMA$                         
\end{tabular}
\caption{\footnotesize Complexity classes in the style of \emph{Merlin-Arthur}. $\QCMA$ is a subset of $\QMA$ where the witness can be efficiently written as a classical bitstring.}
\end{table}

Since any classical bitstring can be efficiently written as a quantum state, $\QCMA$ $\subseteq$ $\QMA$. But is the reverse true? 
Even the \emph{oracle} version of this problem is open: at the top of a recent list of open problems, Aaronson asks for a standard oracle that separates the two classes~\cite{aaronson2021open}.
All previous progress~\cite{Aaronson2007,fefferman2018quantum,nn_2022_classical_randomized} relies on specifically chosen \emph{nonstandardness} in the oracle.

Natarajan and Nirkhe~\cite{nn_2022_classical_randomized} make progress on a standard oracle separation of $\QMA$ and $\QCMA$ by constructing an oracle with randomness. They simultaneously and independently provide a $\QMA$ protocol for testing non-expansion of a graph in an oracle. To prove their lower bound, they combine the adversary method, the polynomial method, and a reduction to a problem of Ambainis, Childs, and Liu~\cite{acl2011}. However, their notion of randomness is different from ours and other works~\cite{harrow_uselessness,fefferman2018quantum,arora22}, and acts as follows: when an oracle is first queried, it chooses a function $f$ from a distribution, but on subsequent queries, it uses the same function $f$. By contrast, our notion of randomness is memoryless: an oracle chooses $f$ from a uniform distribution on $F$ for \emph{every} query. 
This allows one to make small changes to $F$ without affecting the success of the $\QMA$ protocol; we use this flexibility to study a simpler permutation problem.

\subsection{Outline of the paper}
In \Cref{sec:prelim}, we show how to encode the edges of a graph in an invertible function, and define the oracle models and decision problems we consider.
In \Cref{sec:qma}, we explain our $\QMA$ protocol for non-expansion and for a simpler permutation problem of randomized oracles.
We apply representation theory to randomized oracles and prove our classical-witness lower bound in \Cref{sec:classical_witness_lower_bound}; some technical details are deferred to \Cref{appendix:basiselements} and \Cref{appendix:deferred_proofs}.
In \Cref{sec:phase_oracle_lower_bound}, we consider a phase oracle and prove our witness-agnostic lower bound.
We include relevant facts about norms and inner products in \Cref{appendix:norms}, and contrast our setup with quantum walks in \Cref{appendix:quantumwalk}.

\section{Our setup}
\label{sec:prelim}

Consider a $d$-regular graph on $N := 2^n$ vertices for any $n$ and even $d$. We show that an invertible function can list the edges adjacent to each vertex in $G$.
\begin{definition}[Graph-coded function]
Consider a $d$-regular graph $G$ (for even $d$) on $N$ vertices. A $G$-coded function is a function $f: [N] \times [d/2] \to [N]$, such that $f_i(x) := f(x,i)$ is a bijection for each $i \in [d/2]$, and each edge is uniquely represented by a tuple $(x,f_i(x))$.
\end{definition}

\begin{remark}[Even-degree regular graphs have graph-coded functions]
Every regular graph $G$ of even degree has a $G$-coded function.
\end{remark}
\begin{proof}
A $d$-regular graph $G$ of even degree always has a 2-factorization \cite{petersen_graphs}. This means that the edges of $G$ can be partitioned into $d/2$ edge-disjoint subgraphs $[E_1,\dots,E_{d/2}]$ where in each $E_i$, all vertices have degree two (i.e. a collection of cycles). Thus, we can represent each $E_i$ with a permutation $\pi_i$, where the edge $(x,y) \in E_i$ if and only if $\pi_i(x) = y$ or $\pi_i(y) = x$. Then $f(x,i) := \pi_i(x)$ is a $G$-coded function.
\end{proof}

Graph-coded functions $f$ are bijective, and therefore invertible. We now \emph{present} $f$ in various oracle models. Note that we define all oracles with access both to $f$ and $f^{-1}$.
\begin{definition}[An oracle model \emph{presents} a function $f$]
For each oracle model below (e.g. standard oracle), we say that this oracle model \emph{presents} the function $f$.
\end{definition}
\begin{remark}
For notational convenience, we refer to a qubit $z$ that controls the inversion of a function $f$ as taking on values in $\pmset$, so that $f^z$ is either $f^1 = f$ or $f^{-1}$.
\end{remark}
\begin{definition}[Standard oracle]
\label{defn:unitary_std}
For any $f: [N] \to [N]$, define $U_f: \mathbb{C}^{2N^2} \to \mathbb{C}^{2N^2}$ as 
\begin{align}
U_f\sum_{c,x \in [N], z \in \pmset } \alpha_{c,x,z} \ket{c,x,z} := \sum_{c,x \in [N], z \in \pmset} \alpha_{c,x,z} \ket{c \oplus f^z(x), x,z}\,.
\end{align}
\end{definition}

\begin{definition}[In-place oracle \cite{Kashefi_2002}]
\label{defn:unitary_inplace}
For any permutation $\pi: [N] \to [N]$, define $\widetilde{U}_\pi: \mathbb{C}^{2N} \to \mathbb{C}^{2N}$ as 
\begin{align}
    \widetilde{U}_\pi \sum_{x \in [N], z \in \pmset} \beta_{x,z} \ket{x,z} := \sum_{x \in [N], z \in \pmset} \beta_{x,z} \ket{\pi^{z}(x),z}\,.
\end{align}
\end{definition}

\begin{remark}[\cite{Kashefi_2002}]
A standard oracle $U_f$ (with access to $f^{-1}$) can simulate an in-place oracle $\widetilde{U}_f$ in two queries:
\begin{align}
    (\mathbb{I} \otimes X) \circ U_f \circ (SWAP_{n,n} \otimes X)  \circ U_{f} \ket{0}^{\otimes n}\ket{x,z} = \ket{0}^{\otimes n}\ket{f^z(x),z}\,.
\end{align}
\end{remark}

\begin{definition}[$N^\text{th}$ root of unity]
Define the $N^\text{th}$ root of unity as $\omega_N := e^{2\pi i/N}$.
\end{definition}

\begin{definition}[Phase oracle]
\label{defn:unitary_phase_no_control_bits}
For any function $f: [N] \to [N]$, define $\overline{U}_f: \mathbb{C}^{2N} \to \mathbb{C}^{2N}$ as 
\begin{align}
\overline{U}_f\sum_{x \in [N],z \in \pmset } \alpha_{x,z} \ket{x,z} := \sum_{x \in [N], z \in \pmset} \alpha_{x,z} \omega_N^{f^z(x)}\ket{x,z}\,.
\end{align}
\end{definition}

We describe how an oracle in our setup exhibits internal randomness.
On each query, a \emph{randomized} oracle chooses a function uniformly from a set $F$. We say that a randomized oracle \emph{presents} $F$.
\begin{remark}
Given a unitary $U$, we use the notation $\mathcal{U}$ to denote an operator on density matrices; that is,
\begin{align}
    \mathcal{U}[\rho] := U \rho U^\dagger \,.
\end{align}
\end{remark}
\begin{definition}[Randomized oracle (e.g. \cite{harrow_uselessness,fefferman2018quantum})]
For any set $F$ of functions $f: [N] \to [N]$ corresponding to oracles $\{U_f \ |\  f \in F\}$, define the linear operator $\mathcal{O}_F$ as
\begin{align}
    \mathcal{O}_F := \frac{1}{|F|} \sum_{f \in F} \mathcal{U}_f\,.
\end{align}
We match the notation of randomized oracle $\mathcal{O}_F$ with oracle $U_f$; e.g. $\mathcal{\widetilde{O}}_F$ is a randomized \emph{in-place} oracle.
\end{definition}

\subsection{Problem statements}
\label{sec:problem_statement}
The problems below are not fully specified without the choice of oracle model. We prepend the names below with the choice of oracle model; for example, we denote \Cref{problem:simple} in a standard oracle as \problemexpansion{STANDARD}.
\begin{problem}[\problemexpansionnoarg{}]
\label{problem:simple}
Consider an oracle $U_{f}$ presenting a $G$-coded function $f$.  
\begin{enumerate}
    \item In a YES instance, we are promised that $G$ is a union of two disconnected $d$-regular graphs, and that the smaller graph has $N^\alpha$ vertices. 
    \item In a NO instance, we are promised that $G$ is $d$-regular and has spectral gap at least $\epsilon$ (for example, an \emph{expander graph}).
\end{enumerate}
The problem is to decide whether $U_{f}$ is a YES instance or NO instance.
\end{problem}

We also consider a version of this problem with \emph{randomized} oracles, where each randomized YES instance is specified by the set of vertices $V$ of the smaller graph. On each query, an oracle chooses an graph-coded function $f$ uniformly at random that corresponds to a graph where $V$ and $[N]/V$ are disconnected.
\begin{problem}[\problemexpansion{RANDOMIZED}]
\label{problem:simple-randomized}
Consider a randomized oracle $\mathcal{O}_F$ presenting a set of graph-coded functions $F$.
\begin{enumerate}
    \item Each subset $V \subseteq [N]$ of size $|V| = N^\alpha$ specifies a YES instance $\mathcal{O}_{F_V}$. Let $F_V$ be the set of all $G$-coded functions of $d$-regular graphs $G$ with no edges between $V$ and $[N]/V$.
    \item There is a single NO instance $\mathcal{O}_{F_\varnothing}$.
    Let $F_\varnothing$ be the set of all $G$-coded functions of $d$-regular graphs $G$ with spectral gap at least $\epsilon$.
\end{enumerate}
The problem is to decide whether $\mathcal{O}$ is a YES instance or a NO instance.
\end{problem}
In the configuration model of a random graph, $F_V$ contains all functions $f(x,i)$ such that $f_i(x) := f(x,i)$ is the union of a permutation on $[N]/V$ and a permutation on $V$.
In fact, we can use the oracle's internal randomness to adjust the underlying set $F$, and even consider graphs that are not typically expander graphs.
\begin{definition}[Subset indicator]
For a set $V \subseteq [N]$, define the function $i_V: [N] \to \{V, [N]/V\}$ as 
\begin{align}
    i_V(x) = \begin{cases} V & x \in V \\
    [N]/V & x \notin V\,.
    \end{cases}
\end{align}
\end{definition}
\begin{definition}[Permutations that stabilize a subset]
For a set $V \subseteq [N]$, define the set of permutations
\begin{align}
    T_V := \{\pi: [N] \to [N] \, :\, i_V(x) = i_V(\pi(x)) \, \forall x \in [N]\}\,.
\end{align}
We say that $T_V$ \emph{stabilizes} the subset $V$.
\end{definition}

\begin{problem}[\problemhiddensubset{RANDOMIZED}]
\label{problem:hidden_subset}
Consider a randomized oracle $\mathcal{O}_F$ presenting a set of functions $F$.
\begin{enumerate}
    \item  Each subset $V \subseteq [N]$ of size $|V| = N^\alpha$ specifies a YES instance $\mathcal{O}_{T_V}$.
    \item There is a single NO instance $\mathcal{O}_{T_\varnothing}$, where $T_\varnothing$ is the set of all permutations of $[N]$.
\end{enumerate}
The problem is to decide whether $\mathcal{O}$ is a YES instance or a NO instance.
\end{problem}
Notice that \problemhiddensubset{RANDOMIZED} is exactly \problemexpansion[2,\alpha,0]{RANDOMIZED}.

Notice that $T_V$ is a group under function composition. One can generalize this algebraic structure to a problem distinguishing oracles presenting subgroups of $T_\varnothing$ from an oracle presenting $T_\varnothing$:
\begin{problem}[\problemhiddensubgroup{RANDOMIZED}]
\label{problem:hidden_subgroup_randomized}
Consider the set $T_\varnothing$ of all permutations on $[N]$ as a group with operation $\times$, such that each $H_i \subsetneq T_\varnothing$ is also a group. Suppose a randomized oracle $\mathcal{O}$ presents either $T_\varnothing$ or any $H_i$.
\begin{enumerate}
    \item Each subgroup $H_i$ specifies a YES instance $\mathcal{O}_{H_i}$.
    \item There is a single NO instance $\mathcal{O}_{T_\varnothing}$, where $T_\varnothing$ is the set of all permutations of $[N]$.
\end{enumerate}
The problem is to decide whether $\mathcal{O}$ is a YES instance or a NO instance.
\end{problem}
For example, \problemhiddensubset{RANDOMIZED} is a special case of \problemhiddensubgroup{RANDOMIZED} using the group operation of function composition.
\section{Verifying non-expansion with a quantum witness}
\label{sec:qma}
There is a one-query $\QMA$ protocol for \problemexpansionnoarg{} in many oracle models presenting a graph-coded function. 
Graphs with good \emph{expansion} are well-connected despite their sparsity.
For any graph $G$, let $A_G$ be the adjacency matrix of $G$, and $L_G = d\mathbb{I} - A_G$ be the \emph{graph Laplacian} of $G$.
The smallest eigenvalue of $L_G$ is $\lambda_1(L_G) = 0$, and the next-smallest eigenvalue $\lambda_2(L_G)$ measures the expansion of $G$.
In this framework, \problemexpansionnoarg{} asks if an oracle presenting a $G$-coded function has $\lambda_2(L_G) = 0$ (YES), or if $\lambda_2(L_G) \ge \epsilon$ (NO).

At the heart of our protocol is the \emph{spectral test}, which takes an input state $\ket{\psi}$ and fails with probability proportional to $\bra{\psi}L_G\ket{\psi}$. We describe the spectral test for both standard oracles and in-place oracles in \Cref{sub:spectraltest}.
A state that passes the spectral test is essentially supported on a subspace according to $\lambda(L_G) = o(\invpolyn)$; in a NO instance, this is one-dimensional, and in a YES instance, this is at least two-dimensional.
In fact, the uniform superposition over all inputs, $\plus$, is always in this subspace.
As a result, our protocol (\Cref{thm:exists_qma_protocol}) either runs the spectral test, or checks if the input state is close to $\plus$.

Consider the randomized variant of \problemexpansionnoarg{}. The graph of any graph-coded function presented in a YES instance is guaranteed to have a small set $V$ (i.e. $|V| = N^\alpha$) disconnected from the rest of the graph. As a result, there is a state, defined only by the vertices of $V$, that is all-but-negligibly supported in the $\lambda(L_G) = 0$ subspace. This state is the \emph{subset state} $\ket{V}$:
\begin{definition}[Subset state]
\label{defn:subsetstate}
For any non-empty subset $S \subseteq [N]$, define the subset state of $S$ as 
\begin{align}
    \ket{S} := \frac{1}{\sqrt{|S|}} \sum_{x \in S} \ket{x}\,.
\end{align}
\end{definition}
Since $\ket{V}$ is a good witness for \emph{every} graph encoded in a YES instance, the $\QMA$ protocol works just as well in the randomized setting (\Cref{thm:exists_qma_protocol_randomized}).

Randomized oracles that present a set $F$ of graph-coded functions are stable to small changes in the set $F$. In fact, an oracle presenting $F$ encoding all $d$-regular expander graphs is indistinguishable from an oracle presenting $F$ encoding all $d$-regular graphs. The latter oracle can be simulated with $d/2$ queries to the NO instance of \problemhiddensubset{RANDOMIZED}; we show in \Cref{thm:exists_qma_protocol_hiddensubset} that the same $\QMA$ protocol can also decide this problem.

\subsection{The spectral test}
\label{sub:spectraltest}
We give a test that takes an input state $\ket{\psi} = \sum_{x\in[N]} a_x \ket{x}$ on $n$ qubits, and fails with probability proportional to $\bra{\psi}L_G\ket{\psi}$.  This relies on a curious fact:
\begin{lemma}
\label{lemma:curious}
Consider a $d$-regular graph $G$ (for even $d$) on $2^n$ vertices and a $G$-coded function $f$. Suppose we have a normalized quantum state $\ket{\psi} = \sum_{x \in [N]} a_x \ket{x} $ on $n$ qubits. Then 
\begin{align}
\sum_{i\in[d/2]}\sum_{x \in [N]} \| a_x \pm a_{f(x,i)} \|^2
= d \pm \bra{\psi}A_G \ket{\psi}\,.
\end{align}
\end{lemma}
\begin{proof}

\begin{align}
    \sum_{i\in[d/2]} \sum_{x\in[N]}  \| a_x \pm a_{f(x,i)} \|^2
    &= \sum_{i\in[d/2]} \sum_{x\in[N]} \| a_x\|^2 + \| a_{f(x,i)}\|^2 \pm (a_x a_{f(x,i)}^* + a_x^* a_{f(x,i)}) \\
    &= d \pm \sum_{i\in[d/2]} \sum_{x\in[N]} (a_x a_{f(x,i)}^* + a_x^* a_{f(x,i)}) 
    \\
    &= d \pm \bra{\psi} A_G \ket{\psi}.
\end{align}
\end{proof}
We construct the spectral test with one query either to a standard oracle or in-place oracle presenting a graph-coded function $f$. The former (\Cref{procedure:spectraltest_std}) is a SWAP test but with an oracle query in the middle.
The latter (\Cref{procedure:spectral_inplace}) relies on controlled access to the in-place oracle.

\begin{procedure}[Spectral test with a standard oracle]
\label{procedure:spectraltest_std}
Consider a $d$-regular graph $G$ on $N = 2^n$ vertices where $d$ is even, and normalized state $\ket{\psi} = \sum_{x \in [N]} a_x \ket{x} \in \mathbb{C}^{N}$. We assume access to a standard oracle $U_f: \mathbb{C}^{k \times k}$ for $k = N^22^{\lceil \log_2{d} \rceil}$, which acts on a basis vector as
\begin{align}
    U_f\ket{c,x,i,z} = \ket{c \oplus f^z(x,i),x,i,z}\,,
\end{align}
for $c,x \in [N]$, $i \in 2^{\lceil \log_2{d}\rceil - 1}$, and $z \in \pmset$.
\begin{enumerate}
    \item Pick $i \in [d/2]$ uniformly at random, and prepare the state $\ket{i} \in \mathbb{C}^{2^{\lceil \log_2{d} \rceil - 1}}$.
    \item Prepare a qubit in the state $\ket{+} = \frac{\ket{1} + \ket{-1}}{\sqrt{2}} \in \mathbb{C}^2$. (Recall that we label the values of this register in $\pmset$.)
    \item Combine $n$ registers $\ket{0}^{\otimes n}$, the input state $\ket{\psi}$, and $\ket{i}$ and $\ket{+}$ to create $\ket{0}^{\otimes n}\ket{\psi}\ket{i}\ket{+}$.
    \item Apply the oracle $U_f$, which creates the state
    \begin{align}
        \frac{1}{\sqrt{2}} \sum_{x \in [N]} a_x \left(
         \ket{f(x,i)}\ket{x}\ket{i}\ket{1}
         +\ket{f^{-1}(x,i)}\ket{x}\ket{i}\ket{-1} \right) \,.
    \end{align}
    \item Swap the first two sets of $n$ qubits, controlled by the last qubit. This creates the state
    \begin{align}
        &\frac{1}{\sqrt{2}} \sum_{x \in [N]} a_x \left(
         \ket{f(x,i)}\ket{x}\ket{i}\ket{1}
         +\ket{x}\ket{f^{-1}(x,i)}\ket{i}\ket{-1}
         \right)
         \\
         = &\frac{1}{\sqrt{2}} \sum_{x \in [N]} a_x \ket{f(x,i)}\ket{x}\ket{i}\ket{1}
        + 
        \frac{1}{\sqrt{2}} \sum_{x \in [N]} a_{f(x,i)} \ket{f(x,i)}\ket{x}\ket{i}\ket{-1}\,.
    \end{align}
    \item Apply a Hadamard on the last qubit, which creates the state
    \begin{align}
          \frac{1}{2} \sum_{x \in [N]} (a_x + a_{f(x,i)}) \ket{f(x,i)}\ket{x}\ket{i}\ket{1}
        + 
        \frac{1}{2} \sum_{x \in [N]} (a_x - a_{f(x,i)}) \ket{f(x,i)}\ket{x}\ket{i}\ket{-1}\,.
    \end{align}
    \item Measure the last qubit and accept if it is $1$.
\end{enumerate}
Moreover, by \Cref{lemma:curious}, this procedure fails with probability 
\begin{align}
    \frac{1}{d/2}\sum_{i\in[d/2]}  \sum_{x \in [N]} \frac{\| a_x - a_{f(x,i)} \|^2}{4} = \frac{\bra{\psi}L_G\ket{\psi}}{2d}\,.
\end{align}
\end{procedure}

\begin{procedure}[Spectral test with an in-place oracle]
\label{procedure:spectral_inplace}
Consider a $d$-regular graph $G$ on $N = 2^n$ vertices where $d$ is even, and normalized state $\ket{\psi} = \sum_{x \in [N]} a_x \ket{x} \in \mathbb{C}^{N}$. We assume controlled access to an in-place oracle $\widetilde{U}_f: \mathbb{C}^{k \times k}$ for $k = N2^{\lceil \log_2{d} \rceil + 1}$, which acts on a basis vector as
\begin{align}
        \widetilde{U}_f\ket{a,x,i,z} = \ket{a, f^{a \cdot z}(x,i),i,z}\,. 
\end{align}
for control qubit $a \in \{0,1\}$, $x \in [N]$, $i \in 2^{\lceil \log_2{d} \rceil - 1}$, and $z \in \pmset$.\footnote{Note that this procedure does not actually need access to the inverse of $f$ to conduct the spectral test.}
\begin{enumerate}
 \item Pick $i \in [d/2]$ uniformly at random, and prepare the state $\ket{i} \in \mathbb{C}^{2^{\lceil \log_2{d} \rceil - 1}}$.
    \item Prepare a qubit in the state $\ket{+} = \frac{\ket{0} + \ket{1}}{\sqrt{2}} \in \mathbb{C}^2$.
    \item Combine $\ket{+}$, the input state $\ket{\psi}$, $\ket{i}$, and a register $\ket{1}$ to create $\ket{+}\ket{\psi}\ket{i}\ket{1}$.
    \item Apply the oracle $\widetilde{U}_f$, which creates the state
    \begin{align}
        &\frac{1}{\sqrt{2}}\sum_{x \in [N]} a_x \left(\ket{0}\ket{x} + \ket{1}\ket{f(x,i)} \right)\ket{i}\ket{1}  \\
        = &\frac{1}{\sqrt{2}}\sum_{x \in [N]} (a_x \ket{0} + a_{f^{-1}(x,i)} \ket{1})\ket{x}\ket{i}\ket{1}\,.
    \end{align}
    \item Apply a Hadamard on the first qubit, which creates the state
    \begin{align}
         \frac{1}{2}\sum_{x \in [N]} 
         \left(
         (a_x + a_{f^{-1}(x,i)})\ket{0}
         +  (a_x - a_{f^{-1}(x,i)})\ket{1}\right) 
         \ket{x}\ket{i}\ket{1}\,.
            \end{align}
    \item Measure the first qubit and accept if it is $0$.
\end{enumerate}
Moreover, by \Cref{lemma:curious}, this procedure fails with probability 
\begin{align}
    \frac{1}{d/2}\sum_{i \in [d/2]} \sum_{x \in [N]} \frac{\|a_x - a_{f^{-1}(x,i)}\|^2}{4} = \frac{\bra{\psi}L_G\ket{\psi}}{2d}\,.
\end{align}
\end{procedure}

\subsection{A one-query protocol}
\label{sub:onequery}
\begin{theorem}
\label{thm:exists_qma_protocol}
There is a $\QMA$ protocol for \problemexpansion{STANDARD} and \problemexpansion{IN-PLACE} at every even $d \ge 4$, all $0 < \alpha < \frac{1}{2}$, and all constant $\epsilon > 0$.
\end{theorem}
\begin{proof}
Suppose the oracle presents a $G$-coded function. Let $\lambda_1 \le \lambda_2 \le \dots \le \lambda_N$ be the eigenvalues of the graph Laplacian $L_G$. Note that the smallest eigenvalue of a regular graph $G$ is $\lambda_1 = 0$. We always choose the eigenvector associated with $\lambda_1$ as a uniform superposition over vertices of the graph (i.e. $\ket{\lambda_1} := \ket{[N]} = \plus$).

Suppose Arthur receives a state $\ket{\psi} = \sum_{i\in[N]} \alpha_i \ket{\lambda_i}$ from Merlin. Consider the following strategy:
\begin{itemize}
    \item With probability $\frac{1}{2}$, measure $\ket{\psi}$ in the Hadamard basis. Fail if it is in the basis state according to $\plus$, and pass otherwise.
    \item With probability $\frac{1}{2}$, use the spectral test (\Cref{procedure:spectraltest_std} or \Cref{procedure:spectral_inplace}, respectively).
\end{itemize}
The probability of failure $\text{FAIL}$ is 
\begin{align}
 \text{FAIL} &= 
    \frac{1}{2} \|\bra{\psi}\plus \|^2
    + \frac{1}{2} \frac{\bra{\psi} L_G \ket{\psi}}{2d}
    \\
    &= \frac{1}{2} \Big( 
    \|\alpha_1\|^2 + \frac{1}{2d}\sum_{i=1}^N \lambda_i \|\alpha_i\|^2
    \Big)\,.
\end{align}
In a NO instance, $\lambda_k = \Omega(1)$ for all $k > 1$. So the probability of failure is always a positive constant:
\begin{align}
    \text{FAIL}_{\text{NO}} =  \frac{1}{2}\|\alpha_1\|^2 + \frac{1}{2d}\sum_{i=2}^N \Omega(\|\alpha_i\|^2) = \Omega( \sum_{i=1}^N \|\alpha_i\|^2) = \Omega(1)\,.
\end{align}
In a YES instance, the spectrum of $L_G$ is the combined spectrum of the two disconnected graphs. This means $\lambda_1 = \lambda_2 = 0$, and the associated eigenvectors are linear combinations of $\ket{V}$ and $\ket{[N]/V}$.
Recall that $\ket{\lambda_1} := \plus$. We find the orthogonal eigenvector $\ket{\lambda_2}$ in this subspace by inspection:
\begin{align}
    \ket{\lambda_2} = \sqrt{\frac{N-|V|}{N}}\ket{V} +  \sqrt{\frac{|V|}{N}}\ket{[N]/V}\,.
\end{align}
Note that any vector with $\|\alpha_2\|^2 = 1 - o(\invpolyn)$ has negligible probability of failure:
\begin{align}
    \text{FAIL}_{\text{YES}} = \frac{1}{2}
    \|\alpha_1\|^2 + \frac{1}{2d}\sum_{i=1}^N \lambda_i \|\alpha_i\|^2 
    = O(\|\alpha_1\|^2 + \sum_{i=3}^N \|\alpha_i\|^2) = O(1 - \|\alpha_2\|^2)\,.
\end{align}
Suppose Merlin sends the subset state $\ket{V}$. Since $\|\bra{V}\ket{\lambda_2}\|^2 = 1 - \frac{|V|}{N} = 1 - O(\invexpn)$, the strategy has probability of failure $O(\invexpn)$.
\end{proof}

In general, the spectral test can be used in a $\QMA$ protocol to test the magnitude of the second-smallest or largest eigenvalue of a graph Laplacian to inverse polynomial precision. The former is a measure of the quality of a graph's expansion, and the latter is related to a measure of a graph's bipartiteness named the \emph{bipartiteness ratio} \cite{trevisan_bipartitenessratio}.

Because this $\QMA$ protocol requires only one query of either a standard oracle or an in-place oracle, it works even when these oracles are randomized. 
\begin{theorem}
\label{thm:exists_qma_protocol_randomized}
There is a $\QMA$ protocol for \problemexpansion{STANDARD RANDOMIZED} and \problemexpansion{IN-PLACE RANDOMIZED} at every even $d \ge 4$, all $0 < \alpha < \frac{1}{2}$, and all constant $\epsilon > 0$.
\end{theorem}
\begin{proof}
The strategy in \Cref{thm:exists_qma_protocol} also works here. Consider any $G$-coded function presented in a YES instance; the same vertices $V$ are exactly the vertices of the smaller component of $G$. So the witness $\ket{V}$ is close to the second eigenvector $\ket{\lambda_2}$, and the failure probability is negligible. Now consider any $G$-coded function presented in a NO instance. By definition, $G$ is an expander graph, so the failure probability is always a positive constant.
\end{proof}
Because a randomized oracle chooses a function uniformly from a set $F$, it is statistically indistinguishable from an oracle with exponentially small changes to $F$. 
We use this fact to simplify the NO instance in \problemexpansion{RANDOMIZED}. Suppose the NO instance instead presents graph-coded functions of \emph{all} $d$-regular graphs. Since $1 - O(\frac{1}{\poly(N)}) = 1 - O(\invexpn)$ graphs have a constant spectral gap \cite{friedman} when $d > 2$, the failure probability in the $\QMA$ protocol changes by at most $O(\invexpn)$. 

Notice that with this modification, the oracles are exactly $d/2$ copies of the oracles in \problemhiddensubset{RANDOMIZED}.
One way to interpret this is that the randomization offers a substitute for expander graphs.
An expander graph is sparse but well-mixing; a randomized oracle query instantaneously mixes across a graph's connected component. 
As a result, we can distinguish degree-2 graphs with this $\QMA$ protocol, even though they are not typically expander graphs:

\begin{theorem}
\label{thm:exists_qma_protocol_hiddensubset}
There is a $\QMA$ protocol for \problemhiddensubset{STANDARD RANDOMIZED} and \problemhiddensubset{IN-PLACE RANDOMIZED} for all $0 <\alpha < \frac{1}{2}$.
\end{theorem}
\begin{proof}
Perhaps surprisingly, the strategy in \Cref{thm:exists_qma_protocol} also works here:
\begin{itemize}
    \item Consider the graph $G$ of any $G$-coded function presented in a YES instance. By definition, the vertices $V$ are disconnected from all vertices in $[N]/V$. So the witness $\ket{V}$ is close to the second eigenvector $\ket{\lambda_2}$, and the failure probability is negligible.
    \item Consider the NO instance. Then $f$ is chosen uniformly from the set $T_\varnothing$ of all permutations of $[N]$. Then the spectral test fails with probability
    \begin{align}
    \E_{\pi \in T_\varnothing}\left[\frac{d - \bra{\psi} A_\pi \ket{\psi}}{2d} \right]\Bigg|_{d=2}
    &=  \frac{1}{2} - \frac{1}{4} \E_{\pi \in T_\varnothing}\left[{\bra{\psi} A_\pi \ket{\psi}} \right] \\
    &=  \frac{1}{2} - \frac{1}{4}\left( \frac{1}{N!}\sum_{\pi \in T_\varnothing}\bra{\psi}A_\pi \ket{\psi} \right)
    \\
    &= \frac{1}{2} - \frac{1}{8} \left( \frac{1}{(N!)^2}\sum_{\pi_1, \pi_2 \in T_\varnothing}\bra{\psi}(A_{\pi_1} + A_{\pi_2}) \ket{\psi} \right)\,.
    \end{align}
The matrix $A_{\pi_1} + A_{\pi_2}$ determines the adjacency matrix of a random $4$-regular graph in the configuration model; as a result,
\begin{align}
    \E_{\pi \in T_\varnothing}\left[\frac{d - \bra{\psi} A_\pi \ket{\psi}}{2d} \right]\Bigg|_{d=2} = \E_{\pi_1,\pi_2 \in T_\varnothing}\left[\frac{d - \bra{\psi} A_{\pi_1,\pi_2} \ket{\psi}}{2d} \right]\Bigg|_{d=4}\,.
\end{align}
Since a random $4$-regular graph has constant spectral gap with probability $1 - O(\frac{1}{\poly(N)}) = 1 - O(\invexpn)$ \cite{friedman}, the failure probability is at least $\text{FAIL}_{\text{YES}}$ from \Cref{thm:exists_qma_protocol}, less $O(\invexpn)$. So the failure probability is $\Omega(1)$, just as before.
\end{itemize}
\end{proof}

\section{Randomized oracles and symmetric subspaces}
\label{sec:classical_witness_lower_bound}

Our main goal in this section is to show that a general class of verifiers cannot decide \problemhiddensubset{IN-PLACE RANDOMIZED}. 
Recall that in this problem, a verifier has access to a quantum channel and a polynomial-sized classical witness, and must distinguish whether the oracle presents a uniformly random permutation or a permutation that stabilizes a hidden subset $V$. 
Let $\mathcal{Y}$ be the set of all YES instances; note that each instance is uniquely defined by a subset $V$.

Suppose there exists a $\QCMA$ algorithm for this problem.
Since there are at most $O(2^{\poly(n)})$ different classical witnesses, there exists a set of YES instances $\mathcal{Y}'$ that share the same witness, such that $\left| \mathcal{Y}'\right| /\left| \mathcal{Y}\right| = \Omega(2^{-\poly(n)})$.
We can refute the existence of such an algorithm by proving that the same verification ``strategy'' cannot distinguish all instances of $\mathcal{Y}'$ from the NO case with non-negligible probability.
A ``strategy'' is exactly a quantum algorithm: a series of unitaries and oracle queries, followed by a POVM. Without loss of generality, a $T$-query algorithm alternates between unitaries and oracle queries on $\mathcal{H}_O \otimes \mathcal{H}_W$ followed by a measurement\footnote{Note that the last operation does not have to be a unitary -- one can simply replace a unitary followed by a POVM with another equivalent POVM.}, where $\mathcal{H}_O$ is the Hilbert space of the ``oracle'' qubits and $\mathcal{H}_W$ is the extra workspace:
\begin{align}
    \mathcal{E}_O[\rho_0] =  \left(\mathcal{O}\otimes \mathbb{I} \right) \circ \mathcal{U}_{T}\circ \ldots \circ \mathcal{U}_2 \circ\left( \mathcal{O}\otimes  \mathbb{I}\right) \circ \mathcal{U}_1 [\rho_0]\,.
\end{align}
One may try to use the hybrid argument of Bennett, Bernstein, Brassard, and Vazirani~\cite{bbbv} and Ambainis~\cite{Ambainis2000-kd} to prove that the diamond norm $\left|\mathcal{E}_{\oracletv} - \mathcal{E}_{\oracletnone}\right|_\diamond$ is small in expectation over the choice of $\oracletv \in \mathcal{Y}'$. 
This would imply that the verifier cannot distinguish all instances of $\mathcal{Y}'$ with the same strategy. 
We can consider the optimal distinguishing probability in terms of $\left|\mathcal{E}_{\oracletv}[\rho_0] - \mathcal{E}_{\oracletnone}[\rho_0]\right|_1$ for some fixed $\rho_0 \in \mathcal{H}_O \otimes \mathcal{H}_W$.

However, this statistical argument does not hold for some choices of $\mathcal{Y}'$. 
Consider the following simple example: $\mathcal{Y}'$ contains all $V$ such that $1 \in V$. First, $\mathcal{Y}'$ satisfies the size implied by the pigeonhole principle. Second, for $\rho_0 = \ketbra{1} \otimes \mathbb{I}$, $\left|\oracletv[\rho_0] - \oracletnone[\rho_0]\right|_1$ is large for \emph{all} instances in $\mathcal{Y}'$, since $\ket{1}\bra{1}$ mixes only within a small subset.
Note that this only implies the existence of an \emph{instance-specific} POVM distinguishing each YES instance in $\mathcal{Y}'$ from the NO instance.
By contrast, a verification strategy has a \emph{fixed} POVM $\left\{E, \mathbb{I} - E \right\}$.
This allows us to prove that the following value is small on average over the choice of $V$:
\begin{align}
    \left|\Tr\left[E \mathcal{E}_{\oracletv}[\rho_0]\right] - \Tr\left[E \mathcal{E}_{\oracletnone}[\rho_0]\right] \right|
\end{align}
We must bound this value for arbitrary choices of $E$, $\rho_0$ and $\mathcal{U}_i$ fixed in the algorithm. 
In order to do this, we leverage tools from representation theory; this  allows us to see randomized oracles in our problem as \emph{orthogonal projectors} into a subspace of matrices with low dimension.
One caveat of our technique is that 
the verifier is only allowed to have $O(\log(n))$ extra workspace qubits.
This restriction is necessary to reduce the subspace dimension to regimes we can handle.

Representation theory has been previously used to study symmetric operators on variables (in probability) or qubits (in quantum computing) using the language of de Finetti theorems~(e.g. \cite{harrowchurch}); these operators project into subspaces of permutation-invariant sequences or quantum states. By contrast, we notice that some randomized oracles are symmetric operators on \emph{density matrices}. This allows us to explicitly find an orthogonal basis for the associated symmetric subspaces.
We match oracle models with problems with the same group structure: \problemhiddensubset{RANDOMIZED} for in-place oracles in \Cref{sub:classical_witness_notenough}, and an analogous special case of \problemhiddensubgroup{RANDOMIZED} for standard oracles in \Cref{sub:classical_witness_is_enough}.

We now formalize how randomized oracles are orthogonal projectors. We defer the proofs to \Cref{appendix:deferred_proofs}.
\begin{definition}[Representation of a group]
Consider a group $G$ and a vector space $\mathsf{V}$. A \emph{representation} of\, $G$ is a map $R$ that sends each $g \in G$ to a linear operator $R(g): \mathsf{V} \to \mathsf{V}$ such that $R(g_1 g_2) = R(g_1) \circ R(g_2)$  for all $g_1, g_2 \in G$.
\end{definition}

\begin{theorem}[Projecting onto the symmetric subspace {\cite[Proposition 2]{harrowchurch}}]
\label{thm:project_to_sym_subspace}
Consider a finite group $G$, a vector space $\mathsf{V}$, and a representation $R: G \to L(\mathsf{V})$. Then the operator
\begin{align}
    \Pi_R := \frac{1}{|G|} \sum_{g \in G} R(g)
\end{align}
is an orthogonal projector onto $\mathsf{V}^G \subseteq \mathsf{V}$, where
\begin{align}
    \mathsf{V}^G := \{ v \in \mathsf{V} \, :\, R(g)[v]  = v \, \forall g \in G\}\,.
\end{align}
\end{theorem}

\begin{theorem}[Oracles on density matrices form a representation]
\label{thm:oracles_form_rep}
Consider a group $G$ of functions $f: [N] \to [N]$ with bitwise $\oplus$ as the group operation. Then the map $f \mapsto \mathcal{U}_f$ is a representation over the vector space of $2N^2 \times 2N^2$ complex matrices. 

Similarly, consider a group $\widetilde{G}$ of permutations $\pi: [N] \to [N]$ with composition as the group operation. Then the map $\pi \mapsto \mathcal{\widetilde{U}}_\pi$ is a representation over the vector space of $2N \times 2N$ complex matrices.
\end{theorem}

\begin{theorem}[Some randomized oracles are orthogonal projectors]
\label{thm:oracles_are_projectors}
Consider a group $G$ of functions $f: [N] \to [N]$ with bitwise $\oplus$ as the group operation. Then $\mathcal{O}_G$ is an orthogonal projector, under the Frobenius inner product $(x|y) = \Tr[x^\dagger y]$ for $x,y \in \mathbb{C}^{2N^2 \times 2N^2}$, onto
\begin{align}
    \mathsf{V}_G := \{\rho \in \mathbb{C}^{2N^2 \times 2N^2} \, : \mathcal{U}_f[\rho] = \rho\,\forall f \in G\}\,. 
\end{align}
Similarly, consider a group $\widetilde{G}$ of permutations $\pi: [N] \to [N]$ with composition as the group operation. Then $\mathcal{\widetilde{O}}_{\widetilde{G}}$ is an orthogonal projector, under the Frobenius inner product $(x|y) = \Tr[x^\dagger y]$ for $x,y \in \mathbb{C}^{2N \times 2N}$, onto
\begin{align}
    \mathsf{\widetilde{V}}_{\widetilde{G}} := \{\rho \in \mathbb{C}^{2N \times 2N} \, : \mathcal{\widetilde{U}}_\pi[\rho] = \rho\,\forall \pi \in \widetilde{G}\}\,. 
\end{align}
\end{theorem}

In \problemhiddensubset{IN-PLACE RANDOMIZED}, a quantum verifier is either given $\oracletnone$ (NO) or $\oracletv$ for some $V \subseteq [N]$ where $|V| = N^\alpha$ (YES).
Since $T_V \subseteq T_\varnothing$, the symmetric subspace according to $T_\varnothing$ is a subspace of that according to $T_V$, i.e. $\symsubtnone \subseteq \symsubtv$. So we can exactly find the basis of the symmetric subspaces $\symsubtnone$ and $\symsubtv$ (see \Cref{appendix:basiselements} for details). This key property is used throughout~\Cref{sub:classical_witness_notenough}.

\subsection{In-place oracles: when classical witnesses are not enough}
\label{sub:classical_witness_notenough}
We interpret \problemhiddensubset{IN-PLACE RANDOMIZED} as distinguishing the set of all permutations from a subgroup that stabilizes a small subset $V \subseteq [N]$. In \Cref{thm:classicalwitness_notenough}, we prove that classical witnesses designed for the verifier to choose YES cannot help a quantum verifier efficiently decide this problem.
This requires three main lemmas. 
First, we show in \Cref{lemma:good_distinguishers_form_pretty} that input states distinguishing a YES instance or NO instance must have knowledge of the hidden subset $V$ (either as a subset state $\ket{V}$ or a mixed state $\mathbb{I}_V$). 
However, no density matrix can be close to too many subset states $\ket{V}$ (\Cref{lemma:cant_approx_many_subset_states}), and no POVM can choose the right answer for too many mixed states $\mathbb{I}_V$  (\Cref{lemma:cant_all_have_elevated_mean}).
We combine these facts in a hybrid argument; note that we must fix an algorithm by its unitaries \emph{and} its POVM.
We formally state the lemmas (deferring the proofs to~\Cref{appendix:basiselements} and~\Cref{appendix:deferred_proofs}), and then prove \Cref{thm:classicalwitness_notenough}.

We use the following measure of ``progress'' for the hybrid argument:
\begin{definition}[Difference of oracle queries]
\label{defn:difference_dvrho}
For any $\rho$, let $d_{V,\rho}$ be the difference of the two oracle queries
\begin{align}
     d_{V,\rho} := \oracletv[\rho] - \oracletnone[\rho]\,.
\end{align}
\end{definition}
If the nuclear norm of $d_{V,\rho}$ is non-negligible, 
we say that $\rho$ is a good distinguisher of $\oracletv$ and $\oracletnone$. We show that every good distinguisher $\rho$ has a certain form; the proof is deferred to \Cref{appendix:basiselements}.
\begin{lemma}[Good distinguishers have a certain form]
\label{lemma:good_distinguishers_form_pretty}
Consider a density matrix $\rho$ and up to $O(\log(n))$ extra workspace qubits. Suppose $\|d_{V,\rho}\|_1 = \Omega(\invpolyn)$. Then among the quantities
\begin{align}
    \bra{V,z}\rho\ket{V,z}\,,
    \\
    \Tr[ \rho \left( \mathbb{I}_{V,z} - \frac{|V|}{N} \mathbb{I}_{[N],z} \right) ]\,,
\end{align}
for any $z \in \pmset$, at least one has magnitude $\Omega(\invpolyn)$.
\end{lemma}

We now state two lemmas about subsets and subset states. These help us prove that no quantum state can be a good distinguisher of too many YES instances. We defer the proofs to \Cref{appendix:deferred_proofs}.
\begin{lemma}[Can't approximate too many subset states]
\label{lemma:cant_approx_many_subset_states}
Consider a Hermitian $N \times N$ matrix $\rho$ that is positive semidefinite and has trace at most $1$. Consider the set of all subsets $V \subseteq [N]$, where $|V| = N^{\alpha}$ for a fixed $0 < \alpha < \frac{1}{2}$. Then the fraction of subsets $V$ such that $\bra{V}\rho\ket{V} = \Omega(\invpolyn)$ decreases faster than any exponential in $\poly(n)$.
\end{lemma}
\begin{lemma}[Not too many subsets can have elevated mean]
\label{lemma:cant_all_have_elevated_mean}
Consider any $N \times N$ POVM $\{E, \mathbb{I}-E\}$, and the set of all subsets $V \subseteq [N]$, where $|V| = N^\alpha$ for a fixed $0 < \alpha < \frac{1}{2}$. Then the fraction of subsets $V$ where
\begin{align}
    |f(V)| := \left|\frac{1}{|V|}\Tr[\mathbb{I}_V E] - \frac{1}{N}\Tr[E]\right| = \Omega(\invpolyn)\,,
\end{align}
decreases faster than any exponential in $\poly(n)$.
\end{lemma}

Intuitively, \Cref{lemma:cant_approx_many_subset_states} and \Cref{lemma:cant_all_have_elevated_mean} hold because subset states can approximate \emph{any} quantum state well. 
Grilo, Kerenidis, and Sikora~\cite{Grilo15} show that for any $n$-qubit quantum state $\ket{\psi}$, there exists a subset state $\ket{S}$ such that $|\braket{S}{\psi}| \geq \frac{1}{8\sqrt{n+3}}$.

We now prove the main statement:
\begin{theorem}
\label{thm:classicalwitness_notenough}
No quantum verifier that entangles oracle queries with at most $O(\log(n))$ additional qubits can efficiently decide \problemhiddensubset{IN PLACE RANDOMIZED} for any $0 < \alpha < \frac{1}{2}$, even with a polynomial-length classical witness designed for the verifier to choose YES.
\end{theorem}
\begin{proof}
Let the set of YES instances be $\mathcal{Y}$; note that each YES instance corresponds to a set $V \subseteq [N]$ where $|V| = N^\alpha$, for some fixed $0 < \alpha < \frac{1}{2}$.

Suppose for contradiction that there is a protocol for this problem at some $\alpha < \frac{1}{2}$. Then the verifier can distinguish $\oracletnone$ from any $\oracletv$ in a polynomial number of queries using a classical witness of size $O(\poly(n))$.
By the pigeonhole principle, there must exist a set of YES instances $\mathcal{Y}'$ such that $|\mathcal{Y}'|/|\mathcal{Y}| = \Omega(\invexpn)$, where the verifier can use the \emph{same algorithm} to distinguish  $\oracletnone$ from \emph{every} YES instance in $\mathcal{Y}'$. 

We then construct a \emph{hybrid argument} in the style of Bennett, Bernstein, Brassard, and Vazirani~\cite{bbbv} and Ambainis~\cite{Ambainis2000-kd}, which interpolates from queries of one oracle to queries of another oracle. For simplicity we write the proof without extra workspace qubits; however, we can have up to $O(\log(n))$ extra workspace qubits to satisfy \Cref{lemma:good_distinguishers_form_pretty}. Any polynomial query algorithm can be written as a set of unitaries $A = \{U^{(1)},\dots,U^{(k)}\}$ for some $k = O(\poly(n))$ (alternating between unitary evolutions and oracle queries), and a POVM $\{E, \mathbb{I}-E\}$. Consider the following ``hybrid'' algorithms:

\begin{definition}
\label{defn:hybridstate}
Given any set of $k$ unitaries $A = \{U^{(1)},\dots,U^{(k)}\}$, define the hybrid algorithm
\begin{align}
    A_{V, \ell}\left[\rho_0\right] = 
     \oracletv^{(k)}
     \circ
     \mathcal{U}^{(k)}
     \circ
    \dots
    \circ
     \oracletv^{(\ell+1)}
     \circ
      \mathcal{U}^{(\ell)} 
      \circ
       \oracletnone^{(\ell)}
      \circ
       \mathcal{U}^{(\ell)}
    \circ
    \dots 
    \oracletnone^{(1)}
    \circ
    \mathcal{U}^{(1)}\left[\rho_0 \right],
\end{align}
which evolves $\rho_0$ under the oracle $\mathcal{O}_{T_\varnothing}$ for $\ell$ steps and  under $\mathcal{O}_{T_V}$ for the other $k-\ell$ steps.
\end{definition}

 Then the following is true for each $\oracletv \in \mathcal{Y}'$:
\begin{align}
     \Omega(\invpolyn) &=  \left| \Tr[E A_{V,k}[\rho_0]] - \Tr[E A_{V,0}[\rho_0]] 
    \right| 
    \le 
    \sum_{i=0}^{k-1} \left| \Tr[E A_{V,i+1}[\rho_0]] - \Tr[E A_{V,i}[\rho_0]] \right|\,,
    \end{align}
which implies
\begin{align}
    \Omega(\invpolyn) 
    &=  \sum_{i=0}^{k-1} \left| \Tr[E 
    \left(\oracletv^{(k)} \circ
    \dots
        \mathcal{U}^{(i)} \circ
    \left( \oracletv^{(i)}
    - 
     \oracletnone^{(i)}
    \right) [
    \rho^{(i)}]\right)
    ] \right| 
    = \sum_{i=0}^{k-1} \left| \Tr[ E^{V,(i)} d_{V, \rho^{(i)}}] \right|\,,
\end{align}
for the operator $E^{V,(i)}$ constructed by
\begin{align}
    E^{V,(i)} = 
    \mathcal{U}^{\dagger(i)}\circ
    \oracletv^{(i)} 
    \circ
    \dots
    \circ
    \mathcal{U}^{\dagger(k)} \circ
    \oracletv^{(k)} \left[E\right] \,.
\end{align} 
By \Cref{fact:oracle_preserves_properties}, the operators $E^{V,(i)}$ and $\mathbb{I} - E^{V,(i)}$ are also Hermitian and positive semidefinite, so $\{E^{V,(i)}, \mathbb{I}-E^{V,(i)}\}$ is a POVM. 

Using the pigeonhole principle, there must be a step $\ell$ in the summation with magnitude $\Omega(\invpolyn)$. Each $\oracletv \in \mathcal{Y}'$ has such a step. Again by the pigeonhole principle, there is a $\ell^*$ and set $\mathcal{Y}^* \subseteq \mathcal{Y}'$ where
\begin{align}
    \left|
    \Tr[E^{V,(\ell^*)} d_{V,\rho^{(\ell^*)}}]\right|
    = \Omega(\invpolyn)\,,
\end{align} 
and $|\mathcal{Y}^*|/|\mathcal{Y}'| \ge \frac{1}{k} = \Omega(\invpolyn)$. Notice that this implies $|\mathcal{Y}^*|/|\mathcal{Y}| = \Omega(\invexpn)$.

Since the trace of $M$ with a POVM operator is at most $\|M\|_1$ (\Cref{fact:povm_trace_atmost_nuclear_norm}), we have for all $\oracletv \in \mathcal{Y}^*$,
\begin{align}
    \Omega(\invpolyn) = \left| \Tr[E^{V,(\ell^*)} d_{V,\rho^{(\ell^*)}}] \right| \le \left\|d_{V,\rho^{(\ell^*)}}\right\|_1\,.
\end{align}
When queries are entangled with at most $O(\log(n))$ additional qubits, the premise of \Cref{lemma:good_distinguishers_form_pretty} holds; then one of the quantities in the theorem statement must be large.
However, \Cref{lemma:cant_approx_many_subset_states} says that a given $\rho$ can only satisfy either of the first two quantities for a smaller-than-exponential fraction of $\mathcal{Y}$. So for most choices of $\oracletv \in \mathcal{Y}^*$,
\begin{align}
        \Tr[ \rho^{(\ell^*)} \left( \mathbb{I}_{V,z} - \frac{|V|}{N} \mathbb{I}_{[N],z} \right) ] = \Omega(\invpolyn)\,.
\end{align}
for at least one of $z \in \pmset$.

Inspecting the proof of \Cref{lemma:good_distinguishers_form_pretty}, this implies $d_{V,\rho^{(\ell^*)}}$ can only have $\Omega(\invpolyn)$ weight on $C_{4,z}$ for some $z \in \pmset$ across all matrices in $\mathcal{C}$. In fact, for most choices of $\oracletv \in \mathcal{Y}^*$, we show that this is also true for 
\begin{align}
    d_{V,\ell^*,j} := \oracletv^{(\ell^* + j)} \circ
     \mathcal{U}^{(\ell^* + j)} \circ
    \dots
    \circ
    \oracletv^{(\ell^* + 1)} \circ
        \mathcal{U}^{(\ell^* + 1)} \circ
        d_{V,\rho^{(\ell^*)}} 
        \,,
\end{align}
for all $0 \le j \le k-\ell^*$.
We show this by induction. 
Note that by \Cref{fact:povm_trace_atmost_nuclear_norm} and the fact that $d_{V,\ell^*,k-\ell^*}$ is the difference of two objects with nuclear norm $1$, $\|d_{V,\ell^*,k-\ell^*}\|_1 = \Omega(\invpolyn) = O(1)$. 

Consider $d_{V,\ell^*,i}$ for some $1 \le i \le k-\ell^*$, which can be represented with the basis $\mathcal{C}$.
By \Cref{fact:nuclear_doesnt_increase}, it has Frobenius norm at most $\|d_{V, \rho^{(\ell^*)}}\|_{Fr} = O(\frac{1}{\sqrt{|V|}})$.
So it must have $o(\invpolyn)$ weight on pure states. 
Inspecting the basis $\mathcal{C}$, this means $d_{V,\rho^{(\ell^*)}}$ can only have $\Omega(\invpolyn)$ weight on $C_{4,z}$ or $\frac{1}{N}\mathbb{I}_{[N],z}$ for $z \in \pmset$.
By \Cref{fact:nuclear_doesnt_increase}, $d_{V,\rho^{(\ell^*)}}$ has nuclear norm at least $\|d_{V,\ell^*,k-\ell^*}\|_1 = \Omega(\invpolyn)$, so it must have $\Omega(\invpolyn)$ weight on at least one such matrix. 
Suppose for contradiction that the matrix is $\frac{1}{N}\mathbb{I}_{[N],z}$ for $z \in \pmset$. 
Then 
\begin{align}
   \Omega(\invpolyn) = \Tr[\mathbb{I}_{[N],z} \oracletv \left[ \mathcal{U}^{(\ell^*+i)} \left[ d_{V,\ell^*,i-1} \right] \right]] 
   = \Tr[ 
   \left( \mathcal{U}^{(\ell^*+i)\dagger}  \circ
  \oracletv^\dagger[
   \mathbb{I}_{[N],z}
   ]\right)
   d_{V,\ell^*,i-1}]\,.
\end{align}
By the inductive hypothesis, $d_{V,\ell^*,i-1}$ only has $\Omega(\invpolyn)$ weight on some $C_{4,z}$ for $z \in \pmset$. Then for some $z' \in \pmset$,
\begin{align}
 \Omega(\invpolyn) = \Tr[ 
  \left( \mathcal{U}^{(\ell^*+i)\dagger}  \circ
  \oracletv^\dagger[
   \mathbb{I}_{[N],z}
   ]\right)
 \left( \frac{1}{|V|}\mathbb{I}_{V,z'} - \frac{1}{N} \mathbb{I}_{[N],z'} \right) ]\,.
\end{align}
Notice that for any unitary $U$, the object $\{  \mathcal{U}^{(\ell^*+i)\dagger}  \circ
  \oracletv^\dagger[
   \mathbb{I}_{[N],z=+1}
   ],
    \mathcal{U}^{(\ell^*+i)\dagger}  \circ
  \oracletv^\dagger[
   \mathbb{I}_{[N],z=-1}
   ]\}$ forms a POVM. By \Cref{lemma:cant_all_have_elevated_mean}, this can only be satisfied at either $z \in \pmset$ for a smaller-than-exponential fraction of choices of $V$. So for most choices of $\oracletv \in \mathcal{Y}^*$ (i.e. a $\Omega(\invexpn)$ fraction of choices of $V$), $d_{V,\ell^*,i}$ has $\Omega(\invpolyn)$ weight on $C_{4,z}$ for at least one of $z \in \pmset$, and for no other matrices in $\mathcal{C}$.

Since $\Omega(\invpolyn) = \left| \Tr[E d_{V,\ell^*,k-\ell^*}] \right|$, our supposition then implies that for one of $z \in \pmset$,
\begin{align}
 \Omega(\invpolyn) = \left| \Tr[E C_{4,z}] \right| =  \left|\Tr[E \left( \frac{1}{|V|}\mathbb{I}_{V,z} - \frac{1}{N} \mathbb{I}_{[N],z} \right)] \pm O(\invexpn) \right| \,.
\end{align}
But by \Cref{lemma:cant_all_have_elevated_mean}, this can only be satisfied at either $z \in \pmset$ for a smaller-than-exponential fraction of $\mathcal{Y}$. This is a contradiction. So there can be no efficient protocol for this problem.
\end{proof}

\subsection{Standard oracles: when classical witnesses are enough}
\label{sub:classical_witness_is_enough}
As shown in \Cref{thm:oracles_form_rep}, randomized standard oracles can also form a representation. But the preserved group structure is much different than for randomized in-place oracles. Consider the set $T_\varnothing$ of permutations on $[N]$. For any $f_1, f_2 \in T_\varnothing$, the element $f_1f_2$ in this group structure acts for all $x \in [N]$ and $z \in \pmset$ as 
\begin{align}
    (f_1 f_2)^z(x) = f_1^z(x) \oplus f_2^z(x)\,.
\end{align}
Note that this operation is abelian; that is, $ (f_1 f_2) =  (f_2 f_1)$. Any finite abelian group can always be represented as the direct sum of cyclic groups. In fact, under this group operation, $T_\varnothing$ can be decomposed by the input $x \in [N]$ and function inverter $z \in \pmset$:
\begin{align}
    T_\varnothing = \bigoplus_{x \in [2^n], z \in \pmset} \mathbb{Z}_{2^n}\,.
\end{align}
With this group operation, the only possible subgroups of $T_\varnothing$ have the form 
\begin{align}
    \bigoplus_{x \in [2^n], z \in \pmset} \mathbb{Z}_{2^{k_{x,z}}}\,,
\end{align}
for $0 \le k_{x,z} \le n$. As a result, there is a $\QCMA$ protocol to distinguish any strict subgroup of $T_\varnothing$ from $T_\varnothing$.
\begin{theorem}
\label{thm:qcma_alg_std_oracle}
There is a one-query $\QCMA$ protocol for \problemhiddensubgroup{STANDARD RANDOMIZED} when the group operation $\times$ is bitwise XOR, for any valid $\{H_i\}$.
\end{theorem}
\begin{proof}
Suppose the classical witness is a bitstring of length at least $n+1$. The verifier can then:
\begin{enumerate}
    \item Use the first $n$ bits to construct $x$ and the next bit to construct $z$. \item Prepare the state $\ket{0}^{\otimes n}\ket{x,z}$.
    \item Apply $\mathcal{O}_H$, creating the state $\ket{f^z(x)}\ket{x,z}$ for some $f \in H$.
    \item Measure the first $n$ qubits, and accept if the result is even.\footnote{Depending on the encoding, one can simply measure the $n^{\text{th}}$ qubit, and accept if the result is $0$.} 
\end{enumerate}
Consider a YES instance associated with a subgroup $H \subsetneq T_\varnothing$. Then $H$ will have some $x \in [N], z \in \pmset$ such that $k_{x,z} < n$. A witness can store $x$ and $z$; since $k_{x,z} < n$, $f^z(x)$ will be even with probability $1$.

In the NO instance, $H = T_\varnothing$. Then $f^z(x)$ is even with probability $0.5$ for every $x \in [N], z \in \pmset$. 
\end{proof}

Note that \Cref{thm:qcma_alg_std_oracle} holds even if the randomized standard oracle $\mathcal{O}_F$ does not have access to the function inverse.

\section{No witness is enough for phase oracles}
\label{sec:phase_oracle_lower_bound}
We show that deciding \problemhiddensubset{RANDOMIZED}
in a phase oracle is much harder than other oracle models we consider. 
A random phase has \emph{zero} expectation. 
We use this fact to show that queries to most YES instances and the NO instance reduce the magnitude of each off-diagonal of the density matrix by an exponential factor, regardless of the input state.
We bound the Frobenius norm of the difference of query outputs to show that these instances are statistically indistinguishable when the state space is not too large.
As a result, no untrustworthy witness can help decide this problem.

\begin{theorem}
\label{thm:phase_oracle_lower_bound}
No quantum verifier that entangles oracle queries with at most $o(n)$ additional qubits can efficiently decide \problemhiddensubset{PHASE RANDOMIZED} for any $0 < \alpha < \frac{1}{2}$, even with \emph{any} witness designed for the verifier to choose YES.
Moreover, these verifiers require an exponential number of queries to statistically distinguish a YES instance from the NO instance, for \emph{each} of asymptotically all YES instances.
\end{theorem}
\begin{proof}
We first explain why the query lower bound implies that a witness cannot help. In the NO instance, the witness is designed to fool the verifier; in order to overcome this, the verifier must use the witness in tandem with the oracle. But this cannot be done efficiently; regardless of input state, distinguishing the NO instance from nearly any YES instance requires an exponential number of queries.

We now prove the query lower bound. Let $k$ be the number of queries required to distinguish $\mathcal{\overline{O}}_{T_V}$ from $\mathcal{\overline{O}}_{T_\varnothing}$.
Consider any algorithm that distinguishes the two instances, defined by a starting state $\rho_0$, $k$ unitaries, $k$ oracle queries, and a POVM $\{E, \mathbb{I}-E\}$. In the framework of hybrid algorithms (\Cref{defn:hybridstate}),
\begin{align}
    \Omega(\invpolyn) &= \left| \Tr[E A_{V,k}[\rho_0]] - \Tr[E A_{V,0}[\rho_0]] \right|
    \\
    &\le \left\| A_{V,k}[\rho_0] - A_{V,0}[\rho_0] \right\|_1
    \\
    &\le k \max_{i \in \{0,\dots,k-1\}}\left \|A_{V,i+1}[\rho_0] -  A_{V,i}[\rho_0] \right\|_1 
    \\
    &\le k \max_{i \in \{0,\dots,k-1\}}
    \left \|\mathcal{\overline{O}}_{T_V}[\rho^{(i)}] - \mathcal{\overline{O}}_{T_\varnothing}[\rho^{(i)}] \right\|_1 \,,
\end{align}
where the last line follows because randomized oracles do not increase the nuclear norm (\Cref{fact:nuclear_doesnt_increase}).

We now bound $\left\|\mathcal{\overline{O}}_{T_V}[\rho] - \mathcal{\overline{O}}_{T_\varnothing}[\rho] \right\|_1$ for \emph{any} $\rho$. Recall that a phase oracle $\overline{O}_F$ acts as
\begin{align}
    \overline{O}_F\left[\ket{x_1,z_1}\bra{x_2,z_2}\right] = \frac{1}{|F|}\sum_{f \in F} \omega_N^{f^{z_1}(x_1) - f^{z_2}(x_2)}\ket{x_1,z_1}\bra{x_2,z_2}\,,
\end{align}
for any $x_1,x_2 \in [N]$ and $z_1,z_2 \in \pmset$. So every basis vector $\ket{x_1,z_1}\bra{x_2,z_2}$ acquires a coefficient $c_{x_1,z_1,x_2,z_2}$. 

We start with $\overline{O}_{T_\varnothing}$ (the NO instance). When $(x_1,z_1) = (x_2,z_2)$, the coefficient is $1$. When $x_1 \ne x_2$, $f^{z}(x_1)$ and $f^{-z}(x_2)$ are uniformly likely to be any value, so the coefficient is
\begin{align}
     \frac{1}{N^2}\sum_{a \in [N], b \in [N]} \omega_N^{a-b} = \frac{1}{N^2}\left\|\sum_{a \in [N]} \omega_N^a\right\|^2 = 0\,.
\end{align}
Similarly, when $x_1 \ne x_2$, $f^{z}(x_1)$ and $f^{z}(x_2)$ are uniformly likely to be any unequal values; the coefficient is
\begin{align}
     \frac{1}{N(N-1)}\sum_{a \in [N], b \in [N], a \ne b} \omega_N^{a-b} 
    &= \frac{1}{N(N-1)} \big[ \sum_{a \in [N], b \in [N]} \omega_N^{a-b} - \sum_{a \in [N]} \omega_N^{a-a} \big]
    = -\frac{1}{N-1}\,.
\end{align}
We now consider $\overline{O}_{T_V}$ (a YES instance). When $(x_1,z_1) = (x_2,z_2)$, the coefficient is again $1$. When $x_1 \ne x_2$, the values of $f^z(x_1)$ and $f^{-z}(x_2)$ are uniformly likely to be any value in $i_V(x_1)$ and $i_V(x_2)$, respectively, so the coefficient is
\begin{align}
    \frac{1}{|i_V(x_1)| \times |i_V(x_2)|} \sum_{a \in i_V(x_1), b \in i_V(x_2)} \omega_N^{a-b}\,.
\end{align}
Similarly, when $x_1 \ne x_2$, $f^z(x_1)$ and $f^z(x_2)$ are uniformly likely to be any unequal values in $i_V(x_1)$ and $i_V(x_2)$, respectively, so the coefficient is 
\begin{align}
    \frac{1}{|i_V(x_1)| \times |i_V(x_2)|} \sum_{a \in i_V(x_1), b \in i_V(x_2), a \ne b} \omega_N^{a-b} 
    = \frac{
\left(    \sum_{a \in i_V(x_1), b \in i_V(x_2)} \omega_N^{a-b} \right)
-\delta |i_V(x_1)|
    }
    {|i_V(x_1)| \times |i_V(x_2) - \delta|} 
    \,,
\end{align}
where $\delta$ is $1$ if $i_V(x_1) = i_V(x_2)$ and $0$ otherwise.

Consider the object $\mathcal{\overline{O}}_{T_V}[\rho] - \mathcal{\overline{O}}_{T_\varnothing}[\rho]$ as the sum of two matrices $A_V + B_V$.
Let $A_V$ contain the $(V,z) \times (V,z)$ submatrix for both $z \in \pmset$, and $B_V$ contain the rest of the entries. When the oracle query is entangled with $o(n)$ additional qubits, $A_V$ has rank $O(|V|)$, and $B_V$ has rank $O(N)$.

Since the roots of unity sum to zero,
$\sum_{a \in V} \omega_N^a = - \sum_{a \in [N]/V} \omega_N^a$ for any $V \subseteq [N]$.
Because of this,
\begin{align}
    \left \| \sum_{a, b \in i_V(x)} \omega_N^{a-b} \right \| \le \left \| \sum_{a \in i_V(x)} \omega_N^{a} \right \|^2 \le \left \| \sum_{a \in V} \omega_N^{a} \right \|^2 = O(|V|^2)\,.
\end{align}
As a result, all coefficients in $B_V$ are $O(\frac{1}{N^{1-\alpha}})$.

In \Cref{lemma:chernoff}, we show that for asymptotically all choices of $V$, all coefficients in $A_V$ are $O(\frac{1}{N^{3\alpha/4}})$. This argument uses a Chernoff bound and a central limit argument on samples without replacement.

We bound the nuclear norm of $\mathcal{\overline{O}}_{T_V}[\rho] - \mathcal{\overline{O}}_{T_\varnothing}[\rho]$ with the rank and Frobenius norm  of $A_V$ and $B_V$ (\Cref{fact:nuclearnorm_atmost_rank_and_fr}): 
\begin{align}
    \left\| \mathcal{\overline{O}}_{T_V}[\rho] - \mathcal{\overline{O}}_{T_\varnothing}[\rho]\right\|_1 
    &\le \|A_V\|_1 + \|B_V\|_1 
    \\
    &= O(\sqrt{V}) \|A_V\|_{Fr} + O(\sqrt{N})\|B_V\|_{Fr}
    \\
    &\le \left(O(\sqrt{V})O(N^{-3\alpha/4}) + O(\sqrt{N})O(N^{\alpha-1}) \right) \|\rho\|_{Fr}
    \\
    &= O(N^{-\alpha/4} + N^{\alpha-1/2})\,.
\end{align}
Thus, for most choices of $V$, distinguishing $\mathcal{\overline{O}}_{T_V}$ and $\mathcal{\overline{O}}_{T_\varnothing}$ requires $k = \Omega(\min(N^{\alpha/4}, N^{1/2-\alpha}))$ queries.
\end{proof}

We now prove the Chernoff bound:
\begin{lemma}
\label{lemma:chernoff}
Fix any $0 < \alpha < \frac{1}{2}$, and consider all subsets $V \subseteq [N]$ such that $|V| = N^\alpha$. Then for all but a doubly exponentially small fraction of choices of $V$,
\begin{align}
    \left\| \frac{1}{N^{2\alpha}}\sum_{a,b \in V} \omega_N^{a-b} \right\| = O(\frac{1}{N^{3\alpha/4}})\,.
\end{align}
\end{lemma}
\begin{proof}
Consider the distribution $X = \{ \omega_N^k \}$ where $k$ is chosen uniformly from $N$. Both $Re(X)$ and $Im(X)$ have mean zero and variance at most $1$.

Take a size-$N^\alpha$ sample from the distribution $X$, \emph{without replacement}. Denote $Y$ as the distribution of the sample mean. Both $Re(Y)$ and $Im(Y)$ have expectation $Re(X) = Im(X) = 0$, and variance
\begin{align}
    \frac{\sigma_X^2}{N^\alpha}(1 - \frac{N^\alpha - 1}{N - 1}) \le \frac{1}{N^\alpha}\,.
\end{align}
Even when sampling without replacement, $Y$ is asymptotically normally distributed \cite{erdos1959central}. So its moment generating function is
\begin{align}
    \text{MGF}_Y[t] = e^{t\mu_Y + \sigma_Y^2 t^2 / 2} \le e^{t^2/N^\alpha}\,.
\end{align}
We use a Chernoff bound to estimate when $Y$ has magnitude at least $N^{-3\alpha/8}$. Notice that 
\begin{align}
    \Pr[Y \ge a] &= \Pr[e^{tY} \ge e^{ta}] \le e^{-at} \text{MGF}_Y[t] \le e^{t^2/N^\alpha}e^{- at}\,,
\end{align}
so
\begin{align}
    \Pr[Y \ge \frac{0.5}{N^{3\alpha/8}}] 
    &\le
    \inf_{t \ge 0}\  \exp( \frac{t^2}{N^\alpha} - \frac{0.5t}{N^{3\alpha/8}})\ 
    \underset{t = 2N^{\alpha/2}}{\le}
    \exp(4 - N^{\alpha/8}) = O(\frac{1}{\exp(\exp(n))})\,.
\end{align}
This implies that $Y$ has magnitude at most $N^{-3\alpha/8}$ (and $Y^2$ at most $N^{-3\alpha/4}$) except in a doubly exponentially small fraction of choices of $V$.
\end{proof}
\begin{remark}
Consider an oracle that sends $\ket{c,x,z} \to \omega_N^{c \cdot f^z(x)}\ket{c,x,z}$, where the $c$ register has $k$ qubits. 
Note that \Cref{thm:phase_oracle_lower_bound} applies whenever $k = o(n)$.
However, there must be a phase transition, since at $k = n$, this oracle is unitarily equivalent to a standard oracle, and thus has a $\QMA$ protocol for \problemhiddensubset{RANDOMIZED} in \Cref{sec:qma}.
\end{remark}

\section*{Acknowledgements}
\label{sec:ack}
Thanks to Casey Duckering, Juspreet Singh Sandhu, Peter Shor, and Justin Yirka for collaborating on early stages of this project. KM thanks many others for engaging discussions, including 
Adam Bouland, Antares Chen, Aram Harrow, Matt Hastings, Eric Hester, Neng Huang, Robin Kothari, Brian Lawrence,  Yi-Kai Liu, Patrick Lutz, Tushant Mittal, Abhijit Mudigonda, Chinmay Nirkhe, and Aaron Potechin. Thanks to Chinmay Nirkhe for feedback on a previous version of this manuscript.

BF and RB acknowledge support from AFOSR (YIP number FA9550-18-1-0148 and
FA9550-21-1-0008). This material is based upon work partially
supported by the National Science Foundation under Grant CCF-2044923
(CAREER) and by the U.S. Department of Energy, Office of Science,
National Quantum Information Science Research Centers as well as by
DOE QuantISED grant DE-SC0020360.
KM acknowledges support from the National Science Foundation Graduate Research Fellowship Program under Grant No. DGE-1746045. 
Any opinions, findings, and conclusions or recommendations expressed in this material are those of the author(s) and do not necessarily reflect the views of the National Science Foundation.

\bibliography{research.bib}
\bibliographystyle{alpha-betta-url}

\clearpage
\newpage
\appendix

\section{Basis elements of symmetric subspaces}
\label{appendix:basiselements}
In this section we find an orthogonal basis for $\symsubtnone$ and $\symsubtv$, and prove \Cref{lemma:good_distinguishers_form_pretty}.
 
\begin{theorem}[Basis elements of $\symsubtnone$]
\label{thm:basis_elements_fully_symmetric}
Consider $T_\varnothing$, the set of all permutations of $[N]$. Then $\symsubtnone$ is spanned by the matrices
\begin{align}
    A_{z_1, z_2} &=  \plus\braplus \otimes \ket{z_1}\bra{z_2}\,,
    \\
     B'_{z_1} &= \mathbb{I}_{[N]} \otimes \ket{z_1}\bra{z_1}\,,
\end{align}
for $z_1, z_2 \in \pmset$.
\end{theorem}
\begin{proof}
By definition, $\symsubtnone$ contains the set of matrices $M$ such that $U_\pi M U_\pi^\dagger = M$ for all permutations $\pi: [N] \to [N]$. We can always write $M$ in the following form:
\begin{align}
    M &:= \sum_{x_1, x_2 \in [N], z_1, z_2 \in \pmset} \alpha_{x_1,z_1,x_2,z_2} \ket{x_1,z_1}\bra{x_2,z_2}\,.
\end{align}
Then $U_\pi M U_\pi^\dagger$ has the form
\begin{align}
    U_\pi M U_\pi^\dagger &= \sum_{x_1, x_2 \in [N], z_1, z_2 \in \pmset} \alpha_{x_1,z_1,x_2,z_2} \ket{\pi^{z_1}(x_1),z_1}\bra{\pi^{z_2}(x_2),z_2} 
    \\
    &= \sum_{x_1, x_2 \in [N], z_1, z_2 \in \pmset} \alpha_{\pi^{-z_1}(x_1),z_1,\pi^{-z_2}(x_2),z_2} \ket{x_1,z_1}\bra{x_2,z_2}\,.
\end{align}
Thus, for all permutations $\pi$, and for all $x_1, x_2 \in [N]$ and $z_1, z_2 \in \pmset$,
\begin{align}
    \alpha_{x_1, z_1, x_2, z_2} = \alpha_{\pi^{-z_1}(x_1),z_1,\pi^{-z_2}(x_2),z_2}\,.
\end{align}
Given $(x_1,z_1,x_2,z_2)$, this restriction depends on the possible values of $(\pi^{-z_1}(x_1), \pi^{-z_2}(x_2))$:
\begin{itemize}
    \item When $z_1 \ne z_2$, this can take on any value $(j,k) \in [N] \times [N]$.
    \item  When $z_1 = z_2$ and $x_1 \ne x_2$, this can take on any value  $(j,k) \in [N] \times [N]$ such that $j \ne k$.
    \item When $z_1 = z_2$ and $x_1 = x_2$, this can take on any value $(j,j) \in [N] \times [N]$.
\end{itemize}
As a result, for a given choice of $z_1, z_2 \in \pmset$, all coefficients are equal except possibly on the diagonal when $z_1 = z_2$. There are then six basis elements, one proportional to $\mathbb{1}_N \otimes \ket{z_1}\bra{z_2} \propto \plus \braplus \otimes \ket{z_1}\bra{z_2}$ for each $z_1,z_2 \in \pmset$, and one proportional to $\mathbb{I}_{[N]} \otimes \ket{z}\bra{z}$ for each $z \in \pmset$.
\end{proof}
\begin{theorem}[Basis elements of $\symsubtv$]
\label{thm:basis_stabilized_V}
Consider $T_V$ for some non-empty $V \subsetneq [N]$. Then $\symsubtv$ is spanned by the matrices
\begin{align}
    C'_{S_1,S_2,z_1,z_2} &= \ket{S_1}\bra{S_2} \otimes \ket{z_1} \bra{z_2} \,, \\
     D_{S_1,z_1} &= \mathbb{I}_{S_1}  \otimes \ket{z_1}\bra{z_1}\,,
\end{align}
for $S_1,S_2 \in \{V, [N]/V\}$ and $z_1, z_2 \in \pmset$.
\end{theorem}
\begin{proof}
The proof is similar to \Cref{thm:basis_elements_fully_symmetric}.
By definition, $\symsubtv$ contains the set of matrices $M$ such that $U_\pi M U_\pi^\dagger = M$ for all permutations $\pi \in T_V$. We can always write $M$ in the following form:
\begin{align}
    M &:= \sum_{x_1, x_2 \in [N], z_1, z_2 \in \pmset} \alpha_{x_1,z_1,x_2,z_2} \ket{x_1,z_1}\bra{x_2,z_2}\,.
\end{align}
Then for all $\pi \in T_V$, $x_1, x_2 \in [N]$ and $z_1, z_2 \in \pmset$,
\begin{align}
    \alpha_{x_1, z_1, x_2, z_2} = \alpha_{\pi^{-z_1}(x_1),z_1,\pi^{-z_2}(x_2),z_2}\,.
\end{align}
Given $(x_1, z_1, x_2, z_2)$, this restriction depends on the possible values of $(\pi^{-z_1}(x_1), \pi^{-z_2}(x_2))$:
\begin{itemize}
    \item When $z_1 \ne z_2$ or $i_V(x_1) \ne i_V(x_2)$, $\alpha_{x_1,z_1,x_2,z_2}= \alpha_{j,z_1,k,z_2}$ whenever $j \in i_V(x_1)$ and $k \in i_V(x_2)$.
    \item When $z = z_1 = z_2$ and $i_V(x_1) = i_V(x_2)$, $\alpha_{x_1,z,x_2,z}= \alpha_{j,z,k,z}$ whenever $j, k \in i_V(x_1)$ and $j \ne k$.
    \item When $z = z_1 = z_2$ and $x = x_1 = x_2$, $\alpha_{x,z,x,z} = \alpha_{j,z,j,z}$ whenever $j \in i_V(z)$. 
\end{itemize}
As a result, for a given choice of $z_1, z_2 \in \pmset$ and $i_V(x_1),i_V(x_2) \in \{V, [N]/V\}$, all coefficients are equal except possibly on the diagonal when $z_1 = z_2$ and $i_V(x_1) = i_V(x_2)$. There are then twenty basis elements; one proportional to $\ket{i_V(x_1)}\bra{i_V(x_2)} \otimes \ket{z_1}\bra{z_2}$ for each $i_V(x_1),i_V(x_2) \in \{V, [N]/V\}$ and $z_1,z_2 \in \pmset$, and one proportional to $\mathbb{I}_{i_V(x)} \otimes \ket{z}\bra{z}$ for each $i_V(x) \in  \{V, [N]/V\}$ and $z \in \pmset$.
\end{proof}
We can now find an orthogonal basis for $\symsubtnone$ and $\symsubtv$:
\begin{theorem}[Orthogonal basis elements of $\symsubtnone$ and $\symsubtv$]
\label{thm:basis_tv_not_tnothing_orthogonal}
The symmetric subspace according to $T_\varnothing$ is spanned by the orthogonal matrices 
\begin{align}
    A_{z_1,z_2} &= \plus \braplus \otimes \ket{z_1}\bra{z_2}\,,
    \\
    B_{z_1} &= \frac{1}{N}\left(\mathbb{I}_{[N]} - \plus \braplus \right)\otimes \ket{z_1}\bra{z_1}\,,
\end{align}
for $z_1,z_2 \in \pmset$. 

Let $\delta_{i,j}$ be the Kronecker delta function. The symmetric subspace according to $T_V$ for any nonempty $V \subsetneq [N]$ is spanned by the matrices above, and the orthogonal matrices in $\mathcal{C}$, which are
\begin{align}
 C_{1,z_1,z_2} &= \left(\ket{V}\bra{[N]/V} - \ket{[N]/V}\bra{V}\right) \otimes \ket{z_1}\bra{z_2}\,,
    \\
C_{2,z_1,z_2} &= \left( \ket{V}\bra{[N]/V} + \ket{[N]/V}\bra{V} - \frac{2}{N\sqrt{|V|(N-|V|)}} \left(\plus \braplus - \delta_{z_1,z_2}\frac{1}{N} \mathbb{I}_{[N]} \right) \right) \otimes \ket{z_1}\bra{z_2}\,,
\\
    C_{3,z_1,z_2} &= \left( \ket{V}\bra{V} -  \frac{|V|}{N-|V|}\ket{[N]/V}\bra{[N]/V} - \delta_{z_1,z_2}\frac{1-\frac{|V|}{N-|V|}}{N}\mathbb{I}_{[N]} \right) \otimes \ket{z_1}\bra{z_2}\,, \\
    \\
    C_{4,z_1} &= \left( \frac{1}{|V|} (\mathbb{I}_V - \ket{V}\bra{V}) - \frac{1}{N-|V|}(\mathbb{I}_{[N]/V} - \ket{[N]/V}\bra{[N]/V})\right) \otimes \ket{z_1}\bra{z_1}\,,
\end{align}
for $z_1,z_2 \in \pmset$. Moreover, these matrices have Frobenius norm at most $O(1)$ and nuclear norm $\Theta(1)$.
\end{theorem}
\begin{proof}
All matrices are orthogonal when the last qubit $\ket{z_1}\bra{z_2}$ doesn't match. We first consider the basis elements of $\symsubtnone$, as listed in \Cref{thm:basis_elements_fully_symmetric}. When $z_1 \ne z_2$, $A_{z_1,z_2}$ is the only basis vector; when $z = z_1 = z_2$, observe that $B_z$ is a linear combination of $A_{z,z}$ and $B'_z$, and
\begin{align}
    \Tr[A_{z,z}^\dagger B_z] = \frac{1}{N} \left( \braplus \mathbb{I}_{[N]} \plus - \braplus\plus \braplus \plus \right) = 0\,.
\end{align}

We now construct the basis of $\symsubtv/\symsubtnone$ by orthogonalizing the vectors in \Cref{thm:basis_stabilized_V}. When $z_1 \ne z_2$, the only vectors are linear combinations of $C'_{S_1,S_2,z_1,z_2}$ for $S_1,S_2 \in \{V, [N]/V\}$ and $z_1,z_2 \in \pmset$.  Note that $\braplus \ket{V} = \frac{1}{\sqrt{|V|N}}$.
\begin{itemize}
    \item First of all, $C_{1,z_1,z_2}$ is antisymmetric; for every state $\ket{\psi}$, $\bra{\psi} C_{1,z_1,z_2} \ket{\psi} = 0$.
    \item $C_{2,z_1,z_2}$ is constructed as the ``symmetric'' version of $C_{1,z_1,z_2}$, and then offset by $\plus \braplus$ to account for overlap with $A_{z_1,z_2}$.
    \item $C_{3,z_1,z_2}$ is immediately orthogonal $C_{1,z_1,z_2}$ and $C_{2,z_1,z_2}$, and inspection shows orthogonality with $A_{z_1,z_2}$.
\end{itemize}
When $z = z_1 = z_2$, we modify $C_{1,z,z}, C_{2,z,z}, C_{3,z,z}$ to be traceless (i.e. orthogonal to $\mathbb{I}_N \otimes \ket{z}\bra{z}$). The last matrix $C_{4,z}$ must include $\mathbb{I}_V$; it is immediately orthogonal to $C_{1,z,z},C_{2,z,z}$, and  inspection shows that it is traceless and orthogonal to any $\ket{S}\bra{S}\otimes \ket{z}\bra{z}$ for $S \in \{V, [N]/V\}$.

We now calculate the norms of each vector. The Frobenius norm and nuclear norm of an outer product $\ket{\psi}\bra{\psi}$ is exactly $1$. The nuclear norm of $\frac{1}{N}\mathbb{I}_{[N]}$ is $1$, and the Frobenius norm is $\frac{1}{\sqrt{N}}$. Using Cauchy-Schwarz, the Frobenius norm and nuclear norm of $A_{z_1,z_2}, C_{1,z_1,z_2}, C_{2,z_1,z_2}, C_{3,z_1,z_2}$ are all $\Theta(1)$. The nuclear norm of $B_z$ and $C_{4,z}$ are both $\Theta(1)$, and their Frobenius norms are $\Theta(\frac{1}{\sqrt{N}})$ and $\Theta(\frac{1}{\sqrt{|V|}})$, respectively.
\end{proof}

We use the orthogonal basis in \Cref{thm:basis_tv_not_tnothing_orthogonal} to describe the form of good distinguishers of $\oracletv$ and $\oracletnone$.
\begin{fact}[Representing the difference of randomized oracles]
\label{fact:dvrho_form_orthogonal}
Since $\oracletv$ and $\oracletnone$ are both orthogonal projectors under the Frobenius inner product, the difference of their outputs can be represented with the basis elements $\mathcal{C}$ of $\symsubtv/\symsubtnone$ (computed in \Cref{appendix:basiselements}). That is,
\begin{align}
\label{eqn:dvrho_basis}
    d_{V,\rho} := \oracletv[\rho] - \oracletnone[\rho] = \sum_{M \in \mathcal{C}} c_M M \,,
\end{align}
where
\begin{align}
    c_M =  \frac{\Tr[M^\dagger \rho]}{\|M\|_{Fr}^2}\,.
\end{align}
We refer to $c_M$ as the \emph{weight} (of $\rho$) on the matrix $M$.
\end{fact}
\begin{lemma}[\Cref{lemma:good_distinguishers_form_pretty}, restated]
Consider a density matrix $\rho$ and up to $O(\log(n))$ extra workspace qubits. Suppose $\|d_{V,\rho}\|_1 = \Omega(\invpolyn)$. Then among the quantities
\begin{align}
    \bra{V,z}\rho\ket{V,z}\,,
    \\
    \Tr[ \rho \left( \mathbb{I}_{V,z} - \frac{|V|}{N} \mathbb{I}_{[N],z} \right) ]\,,
\end{align}
for any $z \in \pmset$, at least one has magnitude $\Omega(\invpolyn)$.
\end{lemma}
\begin{proof}
By \Cref{fact:dvrho_form_orthogonal}, $d_{V,\rho}$ can be written as a linear combination of basis elements in $\mathcal{C}$, i.e.
\begin{align}
    d_{V,\rho} = \sum_{M \in \mathcal{C}} c_M M = \sum_{M \in \mathcal{C}} \frac{\Tr[M^\dagger \rho] M}{\|M\|_{Fr}^2}\,.
\end{align}
By the pigeonhole principle and the triangle inequality, there is an $M \in \mathcal{C}$ where
\begin{align}
\|c_M M\|_1 \ge \frac{1}{|\mathcal{C}|}\Omega(\invpolyn) = \Omega(\invpolyn)\,,
\end{align}
Note that the last equality holds when $|\mathcal{C}| = O(\poly(n))$; this is true allowing up to $O(\log(n))$ extra workspace qubits.
Inspecting each condition:
\begin{itemize}
    \item Suppose the matrix $M$ is $C_{1,z_1,z_2}$ or $C_{2,z_1,z_2}$ for some $z_1,z_2 \in \pmset$. By the proof of \Cref{thm:basis_tv_not_tnothing_orthogonal}, the Frobenius norm and nuclear norm are both $\Theta(1)$, so
\begin{align}
    \Omega(\invpolyn) = \Tr[C_{1.5 \pm 0.5,z_1,z_2}^\dagger\rho] = \bra{[N]/V,z_2}\rho \ket{V,z_1} \mp \bra{V,z_2}\rho \ket{[N]/V,z_1} \pm O(\invexpn)\,.
\end{align}
By pigeonhole, this implies that there is some $z_1, z_2 \in \pmset$ such that
$
    \left| \bra{[N]/V,z_2}\rho \ket{V,z_1} \right| = \Omega(\invpolyn)
$,
and by \Cref{fact:offdiag_densitymatrix}, $\bra{V,z_1}\rho\ket{V,z_1} = \Omega(\invpolyn)$.
\item  Suppose the matrix $M$ is $C_{3,z_1,z_2}$ for some $z_1,z_2 \in \pmset$. By the proof of \Cref{thm:basis_tv_not_tnothing_orthogonal}, the Frobenius norm and nuclear norm are both $\Theta(1)$, so
\begin{align}
    \Omega(\invpolyn) = \Tr[C_{3,z_1,z_2}^\dagger\rho] = \bra{V,z_2}\rho \ket{V,z_1} \pm O(\invexpn)\,,
\end{align}
and by \Cref{fact:offdiag_densitymatrix}, $\bra{V,z_1}\rho\ket{V,z_1} = \Omega(\invpolyn)$.
\item Suppose the matrix $M$ is $C_{4,z}$ for some $z \in \pmset$. By the proof of \Cref{thm:basis_tv_not_tnothing_orthogonal}, the Frobenius norm is $\Theta(\frac{1}{\sqrt{|V|}})$ and the nuclear norm is $\Theta(1)$, so
\begin{align}
    \Omega(\invpolyn) = |V|\Tr[C_{4,z}^\dagger \rho] = \Tr[\rho \left(\mathbb{I}_{V,z} - \frac{|V|}{N}\mathbb{I}_{[N]/V,z} \right)] - \bra{V,z}\rho\ket{V,z} \pm O(\invexpn)\,.
\end{align}
By pigeonhole, this implies that at least one of the two terms has magnitude $\Omega(\invpolyn)$.
\end{itemize}
This proof holds even if we allow controlled access to the oracle. Consider the addition of a control qubit $\ket{a}$, such that if $a = 0$ the oracle has no effect, and if $a = 1$ the oracle acts as usual.
We can separate $d_{V,\rho}$ into four matrices $d_{V,\rho}^{(a_1,a_2)}$ based on the values $(a_1,a_2)$.

If $\|d_{V,\rho}\|_1 = \Omega(\invpolyn)$, one of the four matrices $d_{V,\rho}^{(a_1,a_2)}$ must have $\Omega(\invpolyn)$ nuclear norm. 
 We have already seen what happens if that matrix is $d_{V,\rho}^{(1,1)}$.
 Recall that the oracle has no effect when $a_1 = a_2 = 0$. Every matrix $M$ with control register $\ket{0}\bra{0}$ satisfies $U_\pi M U_\pi^\dagger = M$ for all permutations $\pi$, and so $M \in \symsubtnone$. As a result, the matrix $d_{V,\rho}^{(0,0)} \in \symsubtv/\symsubtnone$ is identically zero.

Suppose the matrix $d_{V,\rho}^{(1,0)}$ has $\Omega(\invpolyn)$ nuclear norm. (The proof for  $d_{V,\rho}^{(0,1)}$ is nearly identical.) Despite an exponentially large basis of matrices of this form, we show in \Cref{lemma:controlqubitbasis} that this matrix has small rank. In particular, up to $O(\invexpn)$ multiplicative error, $d_{V,\rho}^{(1,0)}$ is a weighted sum of outer products $P_{z=+1} = \ket{1,V,+1}\bra{\psi_{z=+1}}$ and $P_{z=-1} = \ket{1,V,-1}\bra{\psi_{z=-1}}$  for some normalized states $\ket{\psi_z}$.
These outer products have nuclear norm $1$, so by the pigeonhole principle, there is a $z \in \pmset$ such that the weight $c_{P_z}$ on $P_z$ has magnitude at least $O(\invpolyn)$. Since the Frobenius norm of each outer product is also $1$, the weight is $c_{P_z} = \Tr[P_z^\dagger \rho] = \bra{1,V,z}\rho\ket{\psi_z}$. 
Thus, by \Cref{fact:offdiag_densitymatrix}, $\bra{1,V,z}\rho\ket{1,V,z} = \Omega(\invpolyn)$ for some $z \in \pmset$.
\end{proof}

We now prove the lemma:
\begin{lemma}
\label{lemma:controlqubitbasis}
Consider any matrix $A \in \symsubtv/\symsubtnone$ where the control register is $\ket{a_1=1}\bra{a_2=0}$ (or $\ket{a_1=0}\bra{a_2=1}$). Up to $O(\invexpn)$ multiplicative error, $A$ is a linear combination of outer products $\ket{1,V,+1}\bra{\psi_{z=+1}}$ and $\ket{1,V,-1}\bra{\psi_{z=-1}}$ (or their conjugate transposes) for some states $\ket{\psi_{z=+1}}$ and $\ket{\psi_{z=-1}}$.
\end{lemma}
\begin{proof}
We set $a_1 = 1$ and $a_2 = 0$; the proof in the reverse case is nearly identical. Consider any matrix with the form
\begin{align}
      M := \sum_{x_1,x_2 \in [N], z_1,z_2 \in \pmset} \alpha_{x_1,z_1,x_2,z_2} \ket{1,x_1,z_1}\bra{0,x_2,z_2}\,.
\end{align}
For any permutation $\pi$, the matrix $U_\pi M U_\pi^\dagger$ is 
\begin{align}
    U_\pi M U_\pi^\dagger &= \sum_{x_1,x_2 \in [N], z_1,z_2 \in \pmset} \alpha_{x_1,z_1,x_2,z_2} \ket{1,\pi^{z_1}(x_1),z_1}\bra{0,x_2,z_2}
    \\ &= \sum_{x_1,x_2 \in [N], z_1,z_2 \in \pmset} \alpha_{\pi^{-z_1}(x_1),z_1,x_2,z_2} \ket{1,x_1,z_1}\bra{0,x_2,z_2}
    \,.
\end{align}
For every matrix $M \in \symsubtnone$ and permutation $\pi$, $M = U_\pi M U_\pi^\dagger$. Thus, the coefficients do not depend on $x_1$: every $M \in \symsubtnone$ has the form 
\begin{align}
     M = \sum_{x_2 \in [N], z_1,z_2 \in \pmset} \alpha_{z_1,x_2,z_2} \ket{1,+^{\otimes n},z_1}\bra{0,x_2,z_2}
     = \sum_{z \in \pmset}\alpha_{z} \ket{1,+^{\otimes n},z} \bra{\chi_{z}}
     \,.
\end{align}
Similarly, any matrix in $\symsubtv$ satisfies $M = U_\pi M U_\pi^\dagger$ when $\pi$ stabilizes the subset $V$. So if $M \in \symsubtv$, 
\begin{align}
     M = \sum_{S \in \{V, [N]/V\}, x_2 \in [N], z_1,z_2 \in \pmset} \alpha_{S,z_1,x_2,z_2} \ket{1,S,z_1}\bra{0,x_2,z_2}
       = \sum_{S \in \{V, [N]/V\}, z \in \pmset}\alpha_{S,z} \ket{1,S,z} \bra{\chi_{z}}
     \,.
\end{align}
Any matrix in $\symsubtv/\symsubtnone$ must be orthogonal to all matrices in $\symsubtnone$. Thus, if $M \in\symsubtv/\symsubtnone$,
\begin{align}
    M = \sum_{z \in \pmset} \alpha_z \ket{1}\ket{V'}\ket{z} \bra{\psi_{z}}\,,
\end{align}
where
\begin{align}
       \ket{V'} = \sqrt{\frac{N-|V|}{N}}\ket{V} + \sqrt{\frac{|V|}{N}}\ket{[N]/V}\,.
\end{align}
Note that $\ket{V'}$ is equal to $\ket{V}$ up to  $O(\sqrt{\frac{|V|}{N}}) = O(\invexpn)$ multiplicative error.
\end{proof}

\newpage

\section{Norms and inner products}
\label{appendix:norms}
Note that we work with arbitrary matrices, not just positive semidefinite ones.
\begin{definition}[Nuclear norm of a matrix]
\label{defn:nuclearnorm}
The nuclear norm of a matrix $M$ is the sum of its singular values; that is, 
\begin{align}
    \|M\|_1 = \sum_i \sigma_i(M) = \Tr[\sqrt{M^\dagger M}]\,.
\end{align}
\end{definition}
\begin{definition}[Frobenius norm and inner product of a matrix]
The Frobenius inner product of $N \times N$ matrices $A,B$ is 
\begin{align}
    (A|B) = \Tr[A^\dagger B]
\end{align}
This induces a norm, which is the square root of the sum of squares of the singular values:
\begin{align}
    \|A\|_{Fr} = \sqrt{\sum_i \sigma_i(A)^2} = \sqrt{\sum_{ij \in [N]} |A_{ij}|^2 }\,.
\end{align}
\end{definition}
\begin{fact}
\label{fact:nuclearnorm_atmost_rank_and_fr}
The nuclear norm of a matrix is at most the product of its Frobenius norm and the square root of its rank.
\end{fact}
\begin{proof}
See Rennie~\cite{rennie_tracefrobenius} for a proof with explanation.
\end{proof}
\begin{fact}[Nuclear norm of a positive semidefinite matrix]
\label{fact:psd_nuclear_is_trace}
The nuclear norm of a positive semidefinite Hermitian matrix is simply its trace; that is, if $\rho$ is Hermitian and positive semidefinite, then 
\begin{align}
    \|\rho\|_1 = \Tr[\rho]\,.
\end{align}
\end{fact}
\begin{proof}
For a Hermitian and positive semidefinite matrix, the eigenvalues are all real and nonnegative, so the singular values are exactly the eigenvalues. Alternatively, notice that $\rho = \sqrt{\rho^\dagger \rho}$ and use \Cref{defn:nuclearnorm}.
\end{proof}

\begin{fact}[POVM trace is at most the nuclear norm]
\label{fact:povm_trace_atmost_nuclear_norm}
Consider any Hermitian matrix $M$ and a POVM $\{E, \mathbb{I} - E\}$. Then 
\begin{align}
   \Tr[E M]  \le \|M\|_1\,.
\end{align}
\end{fact}
\begin{proof}
Consider the singular value decomposition of a Hermitian $M = UDU^\dagger$. Then
\begin{align}
    \Tr[E M] = \Tr[(U^\dagger E U) D]= \Tr[E' D]
\end{align}
for some matrix $E'$. Note that $\{E', \mathbb{I}-E'\}$ make a POVM; they have the same eigenvalues of $\{E, \mathbb{I}-E\}$, respectively, and so are both positive semidefinite. Recall that the diagonal elements of a POVM are all nonnegative and at most $1$. Then
\begin{align}
     \Tr[E M] = \Tr[E'D] = \sum_{i} E'_{ii} D_i \le \sum_i |D_i| = \|D\|_{1} = \|M\|_1\,.
\end{align}
\end{proof}

\begin{fact}[Trace of outer product is inner product]
\label{fact:outerproduct_vs_innerproduct}
Consider vectors $\ket{x},\ket{y} \in \mathbb{C}^m$ and a matrix $A \in \mathbb{C}^{m \times m}$. Then the inner product of $\ket{y}\bra{x}$ and $A$ is
\begin{align}
    \Tr[(\ket{y}\bra{x})^\dagger A] = \Tr[A\ket{x}\bra{y}] = \bra{y}A\ket{x}\,.
\end{align}
\end{fact}
\begin{proof}
\begin{align}
    \Tr[A \ket{x}\bra{y}] 
    = \Tr[\sum_{i,k \in [m]} \left(\sum_{j\in [m]} A_{ij} x_j \right)_i y^\dagger_k ] = \sum_{k,j \in [m]} A_{kj} x_j y^\dagger_k = \bra{y} A \ket{x}\,.
\end{align}
\end{proof}

\begin{remark}[Orthogonal basis for an input density matrix]
\label{remark:rho_basis}
We can decompose $\rho$ into a basis  $\mathcal{M}$ that is orthogonal under the Frobenius inner product $(a|b) = \Tr[a^\dagger b]$:
\begin{align}
    \rho = \sum_{M \in \mathcal{M}} c_M M\,.
\end{align}
Because the basis is orthogonal, for any $M \in \mathcal{M}$,
\begin{align}
    \Tr[M^\dagger \rho] = \sum_{M' \in \mathcal{M}} c_{M'} \Tr[M^\dagger M'] = c_M \|M\|^2_{Fr}\,.
\end{align}
Moreover, by Cauchy-Schwarz, the inner product of $M$ and $\rho$ is at most the product of the norm of each, so 
\begin{align}
\label{eqn:coeff_upperbound_frob}
    \|c_M M\|_{Fr} = \frac{\left| \Tr[M^\dagger \rho] \right|}{\|M\|_{Fr}} \le \|\rho\|_{Fr}\,.
\end{align}
\end{remark}

We also state two properties that hold for \emph{any} randomized oracle:

\begin{fact}[Randomized oracles preserve trace, Hermiticity, and positive semidefiniteness]
\label{fact:oracle_preserves_properties}
Consider any randomized oracle $\mathcal{O}_F$ corresponding to a set of functions $f \in F$. Then $\mathcal{O}_F$ preserves the trace of its input.
Moreover, if the input $M$ is Hermitian, so is $\mathcal{O}_F[M]$; if $M$ is also positive semidefinite, so is $\mathcal{O}_F[M]$.
\end{fact}
\begin{proof}
Consider any input matrix $M$. Then
\begin{align}
    \Tr[ \mathcal{O}_F[M] ] 
    = \Tr[ \frac{1}{|F|}\sum_{f \in F} \mathcal{U}_f[M] ]  
    = \frac{1}{|F|}\sum_{f \in F} \Tr[ U_f M U_f^\dagger ] = \frac{1}{|F|}\sum_{f \in F} \Tr[ M ] = \Tr[M]\,.
\end{align}
Now suppose $M$ is Hermitian; that is, $M^\dagger = M$. Then
\begin{align}
    \mathcal{O}_F[M]^\dagger 
    = \big( \frac{1}{|F|}\sum_{f \in F} \mathcal{U}_f[M] \big)^\dagger 
    = \frac{1}{|F|}\sum_{f \in F} \big( U_f M U_f^\dagger)^\dagger 
    = \frac{1}{|F|}\sum_{f \in F} U_f M^\dagger U_f = \mathcal{O}_F[M]\,.
\end{align}
Furthermore, suppose $M$ is positive semidefinite; that is, there is a matrix $B$ such that $M = B^\dagger B$. Then 
\begin{align}
    \mathcal{O}_F[M] = \frac{1}{|F|} \sum_{f \in F} \mathcal{U}_f[M] = \frac{1}{|F|} \sum_{f \in F} U_f B^\dagger B U_f^\dagger = \sum_{f \in F} \left( BU_f \right)^\dagger\left( BU_f \right)\,,
\end{align}
which is a sum of positive semidefinite matrices.
Thus, $\mathcal{O}_F[M]$ is positive semidefinite.
\end{proof}

\begin{fact}[Randomized oracles do not increase nuclear norm or Frobenius norm]
\label{fact:nuclear_doesnt_increase}
Consider any randomized oracle $\mathcal{O}_F$ corresponding to a set of functions $f \in F$. Then $\mathcal{O}_F$ does not increase the nuclear norm nor the Frobenius norm of its input.
\end{fact}
\begin{proof}
Recall that both the nuclear norm and Frobenius norm are unitarily invariant. Now consider any input matrix $M$. Then the nuclear norm of $\mathcal{O}_F[M]$ is
\begin{align}
    \left\|\mathcal{O}_F[M]\right\|_1 
    = \left\| \frac{1}{|F|} \sum_{f \in F} U_f M U_f^\dagger \right\|_1 
    \le \frac{1}{|F|} \sum_{f \in F} \left\| U_f M U_f^\dagger \right\|_1
    = \frac{1}{|F|} \sum_{f \in F} \left \| M \right \|_1
    = \left \| M \right \|_1 \,.
\end{align}
The Frobenius norm of $\mathcal{O}_F[M]$ follows in exactly the same way.
\end{proof}
We use one additional property of density matrices in the proof of \Cref{lemma:good_distinguishers_form_pretty}:
\begin{fact}
\label{fact:offdiag_densitymatrix}
    Consider any $N \times N$ density matrix $\rho$ and normalized states $\ket{v},\ket{w}$. If $|\bra{v}\rho\ket{w}| = \Omega(\invpolyn)$, then both $\bra{v}\rho\ket{v}$ and $\bra{w}\rho\ket{w}$ are $\Omega(\invpolyn)$.
\end{fact}
\begin{proof}
Recall that a density matrix is Hermitian and positive semidefinite, so it is diagonalizable and has real and nonnegative eigenvalues. As a result, it has a decomposition
\begin{align}
    \rho = S^\dagger \Lambda S = S^\dagger \sqrt{\Lambda}  \sqrt{\Lambda} S = (\sqrt{\Lambda}S)^\dagger (\sqrt{\Lambda} S) = A^\dagger A\,,
\end{align}
for some diagonal $\Lambda$ and $A := \sqrt{\Lambda} S$. Then by Cauchy-Schwarz,
\begin{align}
    |\bra{v} \rho \ket{w}|^2
    = \left|(A\ket{v})^\dagger  (A\ket{w})\right|^2
    \le \left| (A\ket{v})^\dagger  (A\ket{v})  \right| \cdot \left| (A\ket{w})^\dagger  (A\ket{w}) \right|
    = \bra{v}\rho\ket{v} \cdot \bra{w}\rho \ket{w}\,.
\end{align}
Since $\Tr[\rho] = 1$, $\bra{\psi}\rho\ket{\psi} \le 1$ for all normalized states $\ket{\psi}$.
Thus, both $\bra{v}\rho\ket{v}$ and $\bra{w}\rho\ket{w}$ are at least $|\bra{v} \rho \ket{w}|^2 = \Omega(\invpolyn)$. 
\end{proof}

\newpage
\section{Deferred proofs}
\label{appendix:deferred_proofs}
\begin{theorem}[\Cref{thm:project_to_sym_subspace}, restated; {\cite[Proposition 2]{harrowchurch}}]
Consider a finite group $G$, a vector space $\mathsf{V}$, and a representation $R: G \to L(\mathsf{V})$. Then the operator
\begin{align}
    \Pi_R := \frac{1}{|G|} \sum_{g \in G} R(g)
\end{align}
is an orthogonal projector onto $\mathsf{V}^G \subseteq \mathsf{V}$, where
\begin{align}
    \mathsf{V}^G := \{ v \in \mathsf{V} \, :\, R(g)[v]  = v \, \forall g \in G\}\,.
\end{align}
\end{theorem}
\begin{proof}
We include the proof for completeness.

Note that for any $g \in G$,
\begin{align}
    R(g) \Pi_R = R(g) \frac{1}{|G|} \sum_{g' \in G} R(g') = \frac{1}{|G|} \sum_{g' \in G} R(g g') = \frac{1}{|G|} \sum_{g^{-1} g' \in G} R(g') = \Pi_R\,.
\end{align}
This implies $\Pi_R \Pi_R = \Pi_R$:
\begin{align}
     \Pi_R \Pi_R = \frac{1}{|G|} \sum_{g \in G} R(g) \Pi_R = \frac{1}{|G|} \sum_{g \in G} \Pi_R = \Pi_R\,.
\end{align}
So $\Pi_R$ is a projection.

Note that for any $v \in \mathsf{V}$ and $g \in G$,
\begin{align}
    R(g)[\Pi_R[v]] = (R(g) \circ \Pi_R)[v] = \Pi_R[v]\,,
\end{align}
so $\Pi_R[v] \in \mathsf{V}^G$. And similarly, for all $w \in \mathsf{V}^G$,
\begin{align}
    \Pi_R[w] = \frac{1}{|G|} \sum_{g \in G} R(g)[w] = \frac{1}{|G|}\sum_{g \in G} w = w \in \mathsf{V}^G\,.
\end{align}
So the image of $\Pi_R$ is exactly $\mathsf{V}^G$.

In order to consider orthogonality, we must define an inner product. Consider an inner product $(u|v)$ for $u,v \in \mathsf{V}$. If for some $g \in G$, $(R(g)[u]|R(g)[v]) \ne (u|v)$, use the inner product
\begin{align}
    \langle u,v \rangle 
    = \frac{1}{|G|}\sum_{g \in G} (R(g)[u]|R(g)[v])\,.
\end{align}
Then under this inner product, $R(g)$ is a unitary operator for all $g \in G$:
\begin{align}
    \langle R(g)[u], R(g)[v] \rangle 
    = \sum_{g'g^{-1} \in G} (R(g')[u] | R(g')[v])
    = \langle u, v\rangle \,.
\end{align}
This implies that $\Pi_R$ is an orthogonal projection:
\begin{align}
    \langle \Pi_R[u], v \rangle  
    &= \frac{1}{|G|}\sum_{g,g' \in G} (R(gg')[u]|R(g)[v])  \nonumber
    \\
    &=  \frac{1}{|G|}\sum_{h,h' \in G} (R(h)[u]|R(h')[v]) \nonumber
    \\
    &= \frac{1}{|G|} \sum_{j,j' \in G}  (R(j)[u]|R(jj')[v]) = \langle u, \Pi_R[v] \rangle\,.
\end{align}
\end{proof}

\begin{theorem}[\Cref{thm:oracles_form_rep}, restated; oracles on density matrices form a representation]
Consider a group $G$ of functions $f: [N] \to [N]$ with bitwise $\oplus$ as the group operation. Then the map $f \mapsto \mathcal{U}_f$ is a representation over the vector space of $2N^2 \times 2N^2$ complex matrices. 

Similarly, consider a group $\widetilde{G}$ of permutations $\pi: [N] \to [N]$ with composition as the group operation. Then the map $\pi \mapsto \mathcal{\widetilde{U}}_\pi$ is a representation over the vector space of $2N \times 2N$ complex matrices.
\end{theorem}
\begin{proof}
See that $f_1f_2$ acts as $(f_1 f_2)^{z} (x) = f_1^z(x) \oplus f_2^z(x)$ for any $f_1, f_2 \in G$.
The associated unitary acts as
\begin{align}
    U_{f_1 f_2} \sum_{c,x \in [N], z \in \pmset}\alpha_{c,x,z} \ket{c,x,z} 
    &= \sum_{c,x \in [N], z \in \pmset} \alpha_{c,x,z} \ket{c \oplus (f_1f_2)^z(x), x,z} 
    \\
    &= \sum_{c,x \in [N], z \in \pmset} \alpha_{c,x,z} \ket{c \oplus f_1^z(x) \oplus f_2^z(x), x, z} 
    \\
    &= U_{f_1} U_{f_2} \sum_{c,x \in [N], z \in \pmset}\alpha_{c,x,z} \ket{c,x,z}\,.
\end{align}
Since $U_{f_1f_2} = U_{f_1} U_{f_2}$, $\mathcal{U}_{f_1 f_2} = \mathcal{U}_{f_1} \circ \mathcal{U}_{f_2}$. So the map $f \mapsto \mathcal{U}_f$ respects the group operation of $G$.

Similarly,  $\pi_1 \pi_2$ acts as $(\pi_1 \pi_2)^z(x) = \pi_1^z(\pi_2^z(x))$ for any $\pi_1, \pi_2 \in \widetilde{G}$.
The associated unitary acts as 
\begin{align}
    \widetilde{U}_{\pi_1 \pi_2} \sum_{x \in [N], z \in \pmset} \alpha_{x,z} \ket{x,z}
    &= \sum_{x \in [N], z \in \pmset} \alpha_{x,z} \ket{(\pi_1 \pi_2)^z(x), z} 
    \\ &= \sum_{x \in [N], z \in \pmset} \alpha_{x,z} \ket{\pi_1^z(\pi_2^z(x)), z} 
     = \widetilde{U}_{\pi_1} \widetilde{U}_{\pi_2}  \sum_{x \in [N], z \in \pmset}\alpha_{x,z} \ket{x,z}\,.
\end{align}
Since $\widetilde{U}_{\pi_1 \pi_2} = \widetilde{U}_{\pi_1} \widetilde{U}_{\pi_2}$, $\mathcal{\widetilde{U}}_{\pi_1 \pi_2} = \mathcal{\widetilde{U}}_{\pi_1}\circ  \mathcal{\widetilde{U}}_{\pi_2}$. So the map $\pi \mapsto \widetilde{U}_\pi$ respects the group operation of $\widetilde{G}$.
\end{proof}

\begin{theorem}[\Cref{thm:oracles_are_projectors}, restated; some randomized oracles are orthogonal projectors]
Consider a group $G$ of functions $f: [N] \to [N]$ with bitwise $\oplus$ as the group operation. Then $\mathcal{O}_G$ is an orthogonal projector, under the Frobenius inner product $(x|y) = \Tr[x^\dagger y]$ for $x,y \in \mathbb{C}^{2N^2 \times 2N^2}$, onto
\begin{align}
    \mathsf{V}_G := \{\rho \in \mathbb{C}^{2N^2 \times 2N^2} \, : \mathcal{U}_f[\rho] = \rho\,\forall f \in G\}\,. 
\end{align}
Similarly, consider a group $\widetilde{G}$ of permutations $\pi: [N] \to [N]$ with composition as the group operation. Then $\mathcal{\widetilde{O}}_{\widetilde{G}}$ is an orthogonal projector, under the Frobenius inner product $(x|y) = \Tr[x^\dagger y]$ for $x,y \in \mathbb{C}^{2N \times 2N}$, onto
\begin{align}
    \mathsf{\widetilde{V}}_{\widetilde{G}} := \{\rho \in \mathbb{C}^{2N \times 2N} \, : \mathcal{\widetilde{U}}_\pi[\rho] = \rho\,\forall \pi \in \widetilde{G}\}\,. 
\end{align}
\end{theorem}
\begin{proof}
By \Cref{thm:oracles_form_rep}, the maps $f \mapsto \mathcal{U}_f$ for $f \in G$ and $\pi \mapsto \mathcal{\widetilde{U}}_\pi$ for $\pi \in \widetilde{G}$ both are representations. Note that for any unitary $U$ and square matrices $x,y$ of the same dimensions,
\begin{align}
    (U x U^\dagger | U y U^\dagger) 
    = \Tr[ (U x U^\dagger)^\dagger U y U^\dagger ] 
    = \Tr[ U x^\dagger U^\dagger U y U^\dagger ]
    = \Tr[ U x^\dagger y U^\dagger]
    = \Tr[ x^\dagger y]\,.
\end{align}
So the representation of each group element is a unitary operator under the Frobenius inner product. The proof then follows by \Cref{thm:project_to_sym_subspace}.
\end{proof}

\begin{lemma}[\Cref{lemma:cant_approx_many_subset_states}, restated; can't approximate too many subset states]
Consider a Hermitian $N \times N$ matrix $\rho$ that is positive semidefinite and has trace at most $1$. Consider the set of all subsets $V \subseteq [N]$, where $|V| = N^{\alpha}$ for a fixed $0 < \alpha < \frac{1}{2}$. Then the fraction of subsets $V$ such that $\bra{V}\rho\ket{V} = \Omega(\invpolyn)$ decreases faster than any exponential in $\poly(n)$.
\end{lemma}
\begin{proof}
Let $\mathcal{W}$ be any set of subsets $V$, with size $|\mathcal{W}| = \omega(\poly(n)) = o(\frac{2^{\sqrt{n}}}{\poly(n)})$, such that $|\bra{V_1}\ket{V_2}| = O(2^{-\sqrt{n}})$ for all $V_1, V_2 \in \mathcal{W}$.
Consider an approximate basis for $\rho$, consisting of
\begin{itemize}
    \item $\ket{V}\bra{V}$, for all $V \in \mathcal{W}$, and 
    \item  a set of matrices  $\mathcal{M}'$ orthogonal to all other matrices, and with Frobenius norm $1$.
\end{itemize}
Representing $\rho$ into this approximate basis is not unique, but the coefficient on each term must be close to $\bra{V} \rho \ket{V}$. Recall that the Frobenius norm of a density matrix $\rho$ is at most $1$; by \Cref{remark:rho_basis}, each $|c_M| \le 1$:
\begin{align}
\rho &:= \sum_{V \in \mathcal{W}} c_V \ket{V}\bra{V} + \sum_{M \in \mathcal{M}'} c_M M\,, \\
\bra{V}\rho\ket{V} &= \sum_{V' \in \mathcal{W}} c_{V'} \left|\bra{V'}\ket{V}\right|^2 = c_V + \sum_{V' \in \mathcal{W}, V' \ne V}c_{V'} \left|\bra{V'}\ket{V}\right|^2  = c_V \pm o(\frac{1}{\poly(n)})\,.
\end{align}
Let's inspect $\Tr[\rho^2] \le \Tr[\rho]^2 = 1$:
\begin{align}
    1 \ge \Tr[\rho^2] &= \sum_{V_1 \ne V_2 \in \mathcal{W}} \Tr[ c_{V_1} c_{V_2}^\dagger \ket{V_1}\bra{V_1}\ket{V_2}\bra{V_2}] + \sum_{V \in \mathcal{W}} |c_V|^2 +  \sum_{M \in \mathcal{M}'} |c_M|^2\,.
\end{align}
Notice that the first summation has $o(2^{2\sqrt{n}})$ terms of size $O(2^{-2\sqrt{n}})$, so it has magnitude $o(1)$. Since the last summation is nonnegative, the second summation must be at most a constant $O(1)$. Thus, for at most a $O(\poly(n))$ distinct choices of $V \in \mathcal{W}$, $c_V$ (and thus $\bra{V}\rho\ket{V}$) can have magnitude $\Omega(\invpolyn)$.

So, after choosing a polynomial number of $\ket{V}\bra{V}$ for $V \in \mathcal{W}$, all other subsets $V$ that have good overlap with $\rho$ must have $\omega(2^{-\sqrt{n}})$ overlap with some previously chosen $V \in \mathcal{W}$. The fractional number of subsets $V$ with this property is at most
\begin{align}
\label{eqn:more_than_exp_overlap_more_than_exp_small}
    O(\poly(n)) \times \frac{{N \choose N^{\alpha}(1-\omega(2^{-\sqrt{n}}))}}{{N \choose N^{\alpha}}} = \frac{O(\poly(n))}{N^{ N^\alpha \omega(2^{-\sqrt{n}})}}
    = \frac{1}{2^{n 2^{\alpha n(1 - o(1))}}} = o(\invexpn)\,.
\end{align}
\end{proof}

\begin{lemma}[\Cref{lemma:cant_all_have_elevated_mean}, restated; not too many subsets can have elevated mean]
Consider any $N \times N$ POVM $\{E, \mathbb{I}-E\}$, and the set of all subsets $V \subseteq [N]$, where $|V| = N^\alpha$ for a fixed $0 < \alpha < \frac{1}{2}$. Then the fraction of subsets $V$ where
\begin{align}
    |f(V)| := \left|\frac{1}{|V|}\Tr[\mathbb{I}_V E] - \frac{1}{N}\Tr[E]\right| = \Omega(\invpolyn)\,,
\end{align}
decreases faster than any exponential in $\poly(n)$.
\end{lemma}
\begin{proof}
Suppose for contradiction that $|f(V)| = \Omega(\invpolyn)$ for at least a $\Omega(\invexpn)$ fraction of all subsets $V \subseteq [N]$ where $|V| = N^\alpha$.
Because $f(V)$ is real, it can be either positive or negative; by pigeonhole, a $\frac{1}{2} \Omega(\invexpn) = \Omega(\invexpn)$ fraction of subsets $V$ must have the same sign of $f(V)$ (without loss of generality, assume it is positive). Let $\mathcal{W}$ contain the subsets $V$ with this property. 

Let $k = N^{1-\alpha} - N^{1-\alpha-(\alpha/3-\epsilon)}$ for any $\epsilon > 0$.
We now build a finite sequence $\mathcal{Q} = \{V_1,V_2,\dots, V_k \} \subseteq \mathcal{W}$ that nearly cover $[N]$, and where no two elements share more than $N^{\alpha/3}$ elements. First, choose $V_1 \in \mathcal{W}$. On subsequent draws, choose another $V_m \in \mathcal{W}$ such that
\begin{align}
    \frac{|V_m \cap \left( \cup_{1 \le j < m} V_j \right) |}{|V|} \le N^{\alpha/3}\,.
\end{align}
Suppose we could not draw $V_m \in \mathcal{W}$; then all remaining elements of $\mathcal{W}$ share at least $N^{\alpha/3}$ elements with one of $\{V_1,\dots,V_{m-1}\}$. The fractional size of $\mathcal{W}$ (compared to all subsets) would then be at most
\begin{align}
     (m-1) + \frac
     {
     {(m-1)N^{\alpha} \choose N^{\alpha/3}}
     {N \choose N^{\alpha} - N^{\alpha/3}}
     }
     {{N \choose N^{\alpha}}} 
     \le 
      \frac{
     {(k-1)N^{\alpha} \choose N^{\alpha/3}}
     }
     {{N-(N^\alpha-N^{\alpha/3}) \choose N^{\alpha/3}}} 
    \le  (\frac{k}{N^{1-\alpha} - 1})^{N^{\alpha/3}} \le \left( 1 - \frac{1}{N^{\alpha/3 - \epsilon}} \right)^{N^{\alpha/3}} \le o(\invexpn)\,,
\end{align}
which is a contradiction. So we can always construct $\mathcal{Q}$. Now consider 
\begin{align}
    \sum_{V \in \mathcal{Q}} f(V) = \frac{1}{|V|}\sum_{V \in \mathcal{Q}} \Tr[E \mathbb{I}_V] - \frac{|\mathcal{Q}|}{N}\Tr[E]\,.
\end{align}
By supposition, all $f(V) > 0$ and have magnitude $\Omega(\invpolyn)$, so this quantity should be $\Omega(|\mathcal{Q}|\invpolyn)$. We show however that this sum is small because the two sums are very close. The number of ``overcounts'' of an element  $x \in [N]$ is at most $|\mathcal{Q}|N^{\alpha/3}$, and the number of elements $x \in [N]$ not in some $V \in \mathcal{Q}$ is at most
\begin{align}
    N - k(N^{\alpha} - N^{\alpha/3}) = O(N^{1-(\alpha/3 - \epsilon)})\,.
\end{align}
Recall that the diagonal elements of a POVM are each nonnegative and at most $1$. So then 
\begin{align}
   \left| \frac{1}{|V|}\sum_{V \in \mathcal{Q}} \Tr[E \mathbb{I}_V] - \frac{|\mathcal{Q}|}{N} \Tr[E]\right|
   = O(\frac{|\mathcal{Q}|}{|V|} N^{\alpha/3}) + O(\frac{|\mathcal{Q}|}{N} N^{1-(\alpha/3 - \epsilon)}) = o(\frac{|\mathcal{Q}|}{\poly(n)})\,.
\end{align}
This contradicts our supposition. So the fraction of subsets $V$ where $|f(V)| = \Omega(\invpolyn)$ decreases faster than any exponential in $\poly(n)$.
\end{proof}

\newpage

\section{Our setup contrasted with a discrete-time quantum walk}
\label{appendix:quantumwalk}
 The way one stores a graph in an oracle drastically changes the difficulty of some problems.  Consider a \emph{discrete-time quantum walk} \cite{watrous98}, which allows a vertex access to a superposition of its neighbors.\footnote{\cite[Chapter 17]{childs_notes} has a good introduction to this topic.}
Given a $d$-regular graph $G(V,E)$, the operator $W: \mathbb{C}^{N^2 \times N^2}$ acts as
\begin{align}
    W &= \left( \sum_{(j,k) \in E}\ket{j,k} \bra{k,j} \right) C
    \\
    C &= \sum_{j \in V}\ket{j}\bra{j} \otimes (2 \ket{\partial_j}\bra{\partial_j} - \mathbb{I})
    \\
    \ket{\partial_j} &= \frac{1}{\sqrt{d}}\sum_{(j,k) \in E}\ket{k}\,.
\end{align}
Using a discrete-time quantum walk, we can learn about the mixing properties of the associated graph; these are fundamentally related to the graph's spectral gap \cite{GODSIL2019181}.

By contrast, we query each neighbor of a vertex $v \in G$ with the value of the registers encoding $i \in [d/2]$ (defined by a $G$-coded function). For example, \cite{acl2011} uses a similar oracle to show that deciding whether a graph is a single expander graph or two equal-sized disconnected expander graphs is outside of $\BQP$. Intuitively, a lack of superposition access to neighbors of a vertex makes it harder for a quantum computer to ``traverse'' the graph.

\end{document}